\documentclass[conference]{IEEEtran}

\usepackage{cite}
\usepackage{amsmath,amssymb,amsfonts,amsthm}
\usepackage{algorithm}
\usepackage{algpseudocode}
\usepackage{graphicx}
\usepackage{textcomp}
\usepackage{xcolor}
\usepackage{titlesec}
\usepackage{enumerate}
\usepackage{enumitem}
\usepackage{titlesec}
\usepackage{multirow} 
\usepackage{longtable}

\usepackage{dblfloatfix}

\usepackage[section]{placeins}   
\usepackage{caption}
\usepackage{lscape}

\usepackage{amsthm}
\newtheorem{theorem}{Theorem}[section]
\newtheorem{corollary}[theorem]{Corollary}
\newcommand{\negl}{\mathrm{negl}}

\usepackage{cancel}
\usepackage{ulem} 

\def\BibTeX{{\rm B\kern-.05em{\sc i\kern-.025em b}\kern-.08em
    T\kern-.1667em\lower.7ex\hbox{E}\kern-.125emX}}
\graphicspath{ {./images/} } 

\makeatletter
\algrenewcommand\alglinenumber[1]{\number\value{ALG@line}:}
\makeatother

\renewcommand{\thealgorithm}{\arabic{algorithm}} 

\makeatother

\makeatletter
\newcounter{lineoffset} 
\algrenewcommand\alglinenumber[1]{\number\numexpr#1+\value{lineoffset}\relax:} 
\makeatother

\begin{document}

\title{Privacy-Preserving Patient Identity Management Framework for Secure Healthcare Access}

\author{
\IEEEauthorblockN{Nasif Muslim and Jean-Charles Grégoire}
\IEEEauthorblockA{\textit{Énergie Matériaux Télécommunications Research} \\
\textit{Institut national de la recherche scientifique (INRS)}\\
Montréal, Canada \\
Email: Nasif.Muslim@inrs.ca, jean-charles.gregoire@inrs.ca}
}

\maketitle

\begin{abstract}
Effective healthcare delivery depends on accurate longitudinal health records and addressing patients' concerns regarding the privacy of their information. While patient authentication is essential, reusing patient identifiers exposes individuals to linkability (associating multiple visits) and traceability (tying visits to real-world identities) risks. This paper presents a privacy-preserving, patient-centric identity management framework specifically tailored to the operational and regulatory requirements of healthcare. The framework balances operational reliability with strong privacy protections through a rooted trust anchor, anonymous pseudonyms, and a conditional traceability mechanism. It is formally specified, and its security and privacy properties are evaluated through MSRA-based architectural analysis and complementary formal verification. 
Simulation-based evaluation demonstrates that the framework's identity workflows are operationally feasible within the latency bounds typical of clinical environments.
\end{abstract}


%
\IEEEpeerreviewmaketitle

\vspace{7pt}

\section{Introduction}\label{introduction}

In this paper, the term \textit{patient} refers to any person seeking medical services, preventive care, or access to their health information. Modern healthcare operates on a fundamental tension: delivering high-quality, continuous care requires persistent patient identifiers to link longitudinal records, while protecting patient privacy demands strict limits on how those identifiers are used. In practice, patient identifiers are exchanged across routine healthcare workflows, including registration, appointment check-in, clinician record access, and pharmacy dispensing. Each of these steps depends on timely verification and interoperability between heterogeneous Electronic Health Record (EHR) systems~\cite{miller2004physicians}, often mediated through HL7 FHIR-based interfaces~\cite{bender2013hl7}. To reduce medical errors and support continuity of care, healthcare systems rely on these identifiers. However, the exact identifiers that enable safe record linkage also create privacy risks: they allow activities to be linked across contexts (\textit{linkability}) and tied to a real-world identity (\textit{traceability}).

Processing this patient data is subject to strict legal and privacy obligations. Regulatory frameworks such as the General Data Protection Regulation (GDPR)~\cite{gdpr, voigt2017eu} and the Health Insurance Portability and Accountability Act (HIPAA)~\cite{gostin2009beyond} mandate lawful bases for processing protected health information, enforce principles including data minimization and purpose limitation, and grant individuals rights to access, rectify, and erase their data. These laws reflect a strong societal commitment to privacy and provide mechanisms for holding organizations accountable. However, they primarily define \emph{what} must be achieved legally, rather than \emph{how} it must be technically enforced~\cite{del-real_systematic_2025}. Therefore, legal compliance alone does not automatically guarantee technical privacy properties such as contextual unlinkability or protection against unintended cross-institutional data linkage.

In practice, many healthcare systems implement identity management through institutional perimeter controls and Role-Based Access Control (RBAC) within EHR infrastructures~\cite{bouwman_rights_2008, basil_health_2022}. Once registered, a patient identifier is typically treated as a persistent, long-lived reference used across internal workflows, longitudinal queries, inter-organizational exchanges, and patient portals. This design supports essential functions: (1) authenticating the patient for treatment or services, (2) enabling providers to retrieve historical medical data for safe, continuous care, and (3) allowing patients to access their own records. However, the persistent reuse of the same identifier increases vulnerability to insider threats. It facilitates unintended cross-domain, cross-institutional correlation of sensitive information, running counter to the spirit of purpose limitation and data minimization, even when the initial processing is legally authorized.

An example illustrates both the necessity and the risk of identifiers: A patient visits a community clinic, is examined by a doctor, and receives a prescription. The patient then presents the prescription at a pharmacy to obtain medication.\footnote{Anonymous payments are out of scope in this work.} At the clinic, the healthcare provider must retrieve prior records and update the repository, while at the pharmacy, the pharmacist must confirm that the prescription is dispensed to the intended recipient. Without identifiers, these verifications become unreliable; with heavily reused identifiers, cross-provider privacy is put at risk.

This scenario underscores the dilemma: identifiers are essential for operational reliability, yet their unfettered reuse enables privacy-compromising linkages. Existing approaches to mitigate these risks have proven inadequate. Relying on real-world identifiers (e.g., names, addresses, national IDs) heightens both risks. Healthcare-specific identifiers reduce traceability but still permit cross-provider linkability. Using self-chosen, provider-specific pseudonyms can prevent cross-provider linking, but this approach disrupts operations. Healthcare providers cannot reliably update longitudinal records, and uncoordinated pseudonyms introduce administrative inconsistencies (e.g., duplicate records, overbooking). Thus, neither universal identifiers nor ad hoc pseudonyms achieve the dual goals of privacy and effective record management.

Emerging identity management frameworks proposed in the research literature increasingly utilize Decentralized Identity Management (DIDM) and blockchain-based credential systems to prioritize user control and interoperability. However, these paradigms often under-specify the privacy–continuity tension inherent in healthcare environments. In practice, healthcare systems must reconcile two competing operational requirements. Limiting unnecessary cross-provider linkability during routine interactions reduces the risk of secondary profiling and unintended data correlation. Simultaneously, maintaining clinical safety and regulatory compliance requires mechanisms for lawful traceability and authorized record correlation. Many proposed DIDM-based approaches emphasize user sovereignty but do not formally model or address this balance within healthcare-specific operational constraints.

This paper addresses that gap by presenting a privacy-preserving, patient-centric identity management framework tailored to the unique operational and regulatory requirements of healthcare. The proposed framework integrates established cryptographic mechanisms within a novel system architecture, extending DIDM to meet Healthcare~4.0 priorities, with a particular emphasis on the auditability of healthcare activities. The main contributions of this paper are as follows:

\begin{itemize}

  \item \textbf{Framework:}
    A Healthcare-specific Identity Management (HIDM) framework that reconciles longitudinal record continuity, patient privacy, and regulatory accountability through explicit separation of duties and conditional traceability.

    \item \textbf{System Design and Enforcement Mechanisms:}
    A privacy-preserving identity management architecture in which established cryptographic primitives are employed as unified enforcement mechanisms, rather than as standalone contributions. The design realizes unlinkable authentication, controlled longitudinal record linkage, and conditional traceability through:
    \begin{itemize}
        \item unlinkable credential issuance and selective attribute disclosure for patient authentication based on Camenisch–Lysyanskaya (CL)~\cite{camenisch2004signature} signatures;
        \item privacy-preserving longitudinal record indexing mediated by the HRR using proxy re-encryption (PRE)~\cite{mambo1997proxy};
        \item pseudonym-specific signing key issuance without identity linkage, enabled via blind identity-based signatures (Blind IBS)~\cite{choon2002identity};
        \item secure, authenticated, and replay-resistant interaction tokens issued using Schnorr~\cite{schnorr1991efficient} and partially blind Schnorr~\cite{abe2000provably} signature schemes; and
        \item a separation-of-duties architectural design to enable conditional traceability through distributed authority control, supporting accountability while maintaining baseline patient privacy.
    \end{itemize}

    \item \textbf{Security and Privacy Analysis:}
    A two-layer formal evaluation establishes the robustness of the proposed framework. First, a Multilateral Security Requirements Analysis (MSRA)~\cite{gurses2006contextualizing} systematically maps stakeholder objectives and threat scenarios to the framework's security and privacy controls, demonstrating its alignment with healthcare regulatory principles. Second, a formal cryptographic analysis proves the unforgeability and replay resistance of issuer-signed artifacts under standard computational assumptions, while automated symbolic verification using Scyther and Tamarin verifies the confidentiality, synchronization, integrity, and unlinkability properties under the Dolev–Yao~\cite{dolev2003security} adversary model. Together, these analyzes indicate that the HIDM framework supports patient privacy, provider reliability, and regulatory traceability requirements within the healthcare ecosystem.

  \item \textbf{Empirical Validation:}
    A performance evaluation indicates that the proposed framework is computationally efficient and suggests practical feasibility based on the processing times of core cryptographic operations. The results show that pairing-based CL credentials achieve significantly lower latency than CL--RSA at equivalent security levels across credential issuance, token generation and use, and verification procedures.

\end{itemize}

The remainder of this paper is organized as follows. Section~\ref{sec:background} provides the necessary background and related concepts. Section~\ref{sec:design-overview-healthcare-idm} details the design of the proposed HIDM framework. Section~\ref{sec:security_privacy_analysis} presents the comprehensive security and privacy evaluation, including both MSRA-based and formal analyzes. Section~\ref{sec:performance-measurement} reports the performance evaluation results, and Section~\ref{sec:discussion} concludes the paper with a discussion and directions for future work.

\section{Background} \label{sec:background}

This section provides the conceptual and technical foundations necessary to understand the design of the proposed HIDM framework. 

Section~\ref{sec:identity-management} surveys the landscape of identity management, beginning with the role of credentials, the structure of generic IDM frameworks, and the evolution from traditional centralized models to decentralized, self-sovereign approaches. The discussion concludes by identifying the limitations of existing paradigms when applied to healthcare.

\subsection{Identity Management} \label{sec:identity-management}

This subsection provides an overview of the foundational concepts of identity management, establishing the context for the proposed HIDM framework. Section~\ref{subsubsec:identity-credential-role-idm} introduces the role of identity credentials and clarifies the core terminology. Section~\ref{subsubsec:structure-idm-framework} outlines the generic structure of identity management frameworks, highlighting how entities, credentials, and trust anchors interact. Section~\ref{subsubsec:decentralized-identity-management} discusses the shift from traditional centralized and federated models toward decentralized, self-sovereign identity paradigms. Finally, Section~\ref{subsubsec:limitations-idm-healthcare} identifies the limitations of existing approaches in the healthcare domain, motivating the need for specialized, domain-specific solutions.

\subsubsection{Identity, Credential, and Their Role in Identity Management}
\label{subsubsec:identity-credential-role-idm}

An \textit{identity} is a structured set of attributes that uniquely distinguishes an entity within a specific context. 
For example, an individual's identity may comprise their name, date of birth, residential address, and occupation. 
In digital ecosystems, establishing trust between entities essential for secure information exchange requires \textit{authentication}, whereby one party proves a claimed identity to another.

A \textit{credential} is an artifact that provides verifiable evidence of a claimed identity. 
Credentials can be broadly classified into two categories: physical credentials and digital credentials~\cite{adkins2020building}.

\begin{itemize}

  \item \textbf{Physical Credentials}: Tangible artifacts that provide proof of identity and personal details, typically issued by trusted authorities such as government agencies. Examples include identification cards and passports. In healthcare, these credentials are often used during patient onboarding to confirm a patient's identity and verify eligibility for services.

  \item \textbf{Digital Credentials}: Intangible artifacts that prove identity in online or electronic transactions. They are generally divided into three categories:
  
    \begin{itemize}

      \item \textbf{Knowledge-Based}: Verified using information known only to the user, such as passwords, PINs, or answers to security questions.
      
      \item \textbf{Possession-Based}: Verified through something the user possesses, such as public key certificates (e.g., X.509), attribute-based credentials (ABCs), one-time passwords (OTPs), smart cards, or hardware tokens.
      
      \item \textbf{Property-Based (Biometrics)}: Verified through inherent characteristics of the user, such as physiological traits (e.g., fingerprints, facial features, iris scans, DNA, vein patterns) and behavioral patterns (e.g., keystroke dynamics, mouse usage, gait recognition).
      
    \end{itemize}
    
\end{itemize}

In the latter case, biometric traits are intrinsically linked to the individual and cannot be easily rotated or revoked in the same manner as passwords or cryptographic keys. Traditional storage of raw biometric data poses significant security risks, including unauthorized access, large-scale data breaches, and identity misuse. Biometric Template Protection (BTP) schemes~\cite{dong2019risk} mitigate these risks by transforming biometric data into protected templates parameterized by secret or system-specific values, such that recovering the original biometric sample from the protected template is computationally infeasible under standard assumptions. Examples include Locality Sensitive Hashing (LSH)~\cite{indyk1998approximate} and Bio-hashing~\cite{kong2006palmprint}. These probabilistic transformations enable matching of similar (but not identical) biometric samples, an essential property for noisy inputs, while reducing invertibility and cross-template linkability. If a protected template is compromised, a new unlinkable template can be generated by reapplying the transformation with fresh parameters.

\subsubsection{An IDM Framework}
\label{subsubsec:structure-idm-framework}

An Identity Management (IDM) framework defines the functions required to govern both the lifecycle and usage of credentials. 
These functions are organized into two core modules: the Credential Issuance and Maintenance Module (CIMM) and the Access Management Module (AMM)~\cite{grassi2017digital}. 
The modular design enables flexibility and allows adaptation across different identity management paradigms. 
Each module and its functions are detailed below:

\paragraph{Credential Issuance and Maintenance Module (CIMM)}  
The CIMM manages all phases of a credential's lifecycle through a set of processes that are detailed below:

\begin{itemize}
    \item \textbf{Credential Generation} Validates attributes, checks documents, and issues credentials.
    \item \textbf{Status Verification} Checks that credentials remain valid and policy-compliant.
    \item \textbf{Update} Refreshes credential data when attributes change.
    \item \textbf{Revocation} Removes credentials if they are compromised, expired, or requested.
\end{itemize}

\paragraph{Access Management Module (AMM)}  
The AMM governs how credentials are used and enforces access policies for services and resources. The functions of the AMM are detailed below:

\begin{itemize}
    \item \textbf{Authentication} Verifies the entity’s credentials to confirm it is who it claims to be, for example, by checking a PIN, password, or biometric data linked to the credential identifier.
    \item \textbf{Authorization} Determines the access level granted to the entity based on its role, the applicable policies, and the context of the request.
    \item \textbf{Accounting} Maintains detailed audit logs of resource access, including who accessed what, when, and which actions were taken. It supports accountability, compliance audits, and the detection of suspicious activities.
\end{itemize}

\begin{table*}[!b] 
\centering
\caption{Comparison of Blockchain-Based Healthcare Identity Frameworks with Respect to Identifier Exposure, Longitudinal Support, and Traceability}
\label{tab:health-id-comparison}
\begin{tabular}{|p{2.0cm}|p{3.3cm}|p{3.3cm}|p{3.3cm}|p{3.3cm}|}
\hline
\textbf{Framework} &
\textbf{Persistent Identifier Exposed?} &
\textbf{Supports Longitudinal Records} &
\textbf{Provider-Unlinkability} &
\textbf{Conditional Traceability} \\
\hline

Health-ID (Javed et al. \cite{javed2021health}) &
\textbf{Yes} -- ledger-visible global identifier (\texttt{healthID}) &
\textbf{Partial} -- identity continuity without privacy-preserving record evolution &
\textbf{No} -- same identifier reused across domains &
\textbf{No} -- regulator-centric revocation only \\
\hline

Zhuang et al. \cite{zhuang2020patient} &
\textbf{Yes} -- ledger-visible global blockchain identifier &
\textbf{Yes} -- touchpoint-indexed longitudinal visits &
\textbf{No} -- global identifier reused across facilities &
\textbf{No} -- unconditional auditability without selective disclosure \\
\hline

Bai et al. \cite{bai2022self} &
\textbf{Yes} -- ledger-visible DIDs and metadata &
\textbf{Partial} -- identity continuity without privacy-preserving record evolution &
\textbf{Partial} -- pairwise DIDs, but ledger-level metadata correlation &
\textbf{No} -- revocation without selective deanonymization \\
\hline

Torongo et al. \cite{torongo2023blockchain} &
\textbf{Yes} -- ledger-visible DIDs and audit metadata &
\textbf{Partial} -- identity continuity without privacy-preserving record evolution &
\textbf{Partial} -- pairwise DIDs, but ledger-level correlation &
\textbf{No} -- full auditability without policy-bound tracing \\
\hline

\end{tabular}
\end{table*}

\subsubsection{Decentralized Identity Management (DIDM)}
\label{subsubsec:decentralized-identity-management}

Trust in digital identity management can be understood at two levels:
Level 1 — Cryptographic Authentication, where an entity proves control of a private key corresponding to a public key, and
Level 2 — Real-World Legitimacy, where the entity is recognized as a genuine legal or organizational actor.
In the Traditional Identity Management (TIDM) paradigm~\cite{rfc5280}, trust is anchored in centralized infrastructures. 
Traditional TLS-based channel authentication of the TIDM addresses both cryptographic authentication and partial real-world legitimacy (domain ownership and, in some cases, organizational identity) through handshakes based on public key certificates. These certificates are issued by Certificate Authorities (CAs) of the Public Key Infrastructure (PKI), which act as global trust anchors. The primary limitation of the TIDM is its reliance on central authorities at both the cryptographic and real-world legitimacy levels, which creates bottlenecks and single points of failure.

The DIDM paradigm~\cite{muhle2018survey} removes reliance on centralized PKI and Certificate Authorities (CAs) for cryptographic trust. Instead, it employs a Decentralized Public Key Infrastructure (DPKI)~\cite{Papageorgiou2020DPKI}, where key ownership is anchored in a distributed, tamper-evident ledger. 
This ledger stores Decentralized Identifiers (DIDs)~\cite{w3c-did-core} and their associated public keys, enabling any party to verify key ownership without relying on a central CA.
Unlike traditional PKI certificates, a DID functions only as an identifier (it proves control of a private key corresponding to a published public key) and does not establish real-world legitimacy (e.g., a licensed hospital). Such legitimacy must instead be conveyed through Verifiable Credentials (VCs) issued by recognized authorities and cryptographically bound to the holder's DID.

The following features characterize the DIDM paradigm:

\begin{itemize}
    \item \textbf{Possession-based credentials:} Typically implemented as attribute-based credentials that are digitally signed by the issuer and held by the user.
    \item \textbf{Local credential verification:} Verifiers independently check the cryptographic signature without requiring real-time validation from the issuer or a trusted third party.
\end{itemize}

Building on these features, the DIDM paradigm underpins the Self-Sovereign Identity (SSI) model~\cite{allen2016path}, which advances a user-centric approach to digital identity. 
SSI emphasizes individual control, informed consent, and privacy, ensuring that users remain the ultimate custodians of their identity data. 
Its design is further shaped by foundational privacy and usability principles, notably those articulated by Cameron in the Laws of Identity~\cite{cameron2005laws}:


These principles are implemented in SSI systems through three core components: Decentralized Identifiers (DID), DID-based Messaging (DIDComm), and Verifiable Credentials (VCs).

\begin{itemize}

    \item \textbf{DID}:
    A globally unique, persistent identifier for a digital entity that does not require a centralized registry~\cite{reed2020decentralized}. 
    Each DID resolves to a DID Document, which specifies public keys, verification methods, and service endpoints necessary for secure interactions with the entity. 

    %
    
    \item \textbf{DIDComm}:
    An application-layer protocol that enables secure, transport-agnostic communication between parties identified by DIDs~\cite{curren2023didcomm}. Unlike TLS, which secures a single connection via handshake and CA-issued certificates, it provides end-to-end encryption and mutual authentication by encrypting and signing each message with keys from the participants' DID Documents.
    This design supports secure routing through untrusted intermediaries (relays), asynchronous communication even when parties are offline, and compatibility with any transport (e.g., HTTP, Bluetooth).
    
    \item VC:
    A digitally signed, tamper-evident artifact that encapsulates a set of claims made by an Issuer about a Subject ~\cite{model2022verifiable}. Within DIDM, VCs serve as the primary mechanism for conveying real-world legitimacy and attributes. 
    The core elements are:
    \begin{itemize}
        \item Issuer: The trusted entity that generates and signs the credential.
        \item Subject: The entity identified by a DID to which the claims apply.
        \item Claims: Assertions such as age, role, or qualifications.
        \item Proof: A cryptographic signature that ensures authenticity and integrity.
    \end{itemize}

\end{itemize}

\subsubsection{Limitations of Identity Management in Healthcare} 
\label{subsubsec:limitations-idm-healthcare}

Standard deployments of traditional IDM paradigms lack the privacy-preserving and user-controlled features required in patient-centric healthcare ecosystems. 
While the SSI model within the DIDM paradigm addresses some of these gaps, it remains insufficient for the specialized demands of healthcare identity management. 
Regulatory frameworks such as HIPAA~\cite{hipaa} and GDPR~\cite{gdpr} impose strict requirements that expose the following shortcomings:

\begin{enumerate}
    \item \textbf{Lack of Verifiable Legitimacy:}  
     DID and DID Documents verify control over a cryptographic key, but they do not by themselves establish an entity’s real-world legitimacy (e.g., whether a clinic or pharmacy is legally authorized to provide healthcare services). Although provider credential verification is often framed as a security or anti-fraud concern, it is also directly relevant to privacy: disclosing sensitive medical data or cryptographic proofs to an unauthenticated or impersonated endpoint may result in unauthorized data exposure. Therefore, patients require mechanisms to verify that they are interacting with licensed and legitimate healthcare providers (e.g., clinics, hospitals, pharmacies)~\cite{preukschat2021self}. 
    While standard TLS (X.509-based PKI)~\cite{ cooper_internet_2008} provides transport-layer authentication and channel security, it does not inherently encode healthcare-specific authorization or legitimacy attributes beyond domain or organizational identity. Equipping healthcare providers with DIDs and legitimacy-verifiable credentials bridges this gap, allowing cryptographically secure, application-layer verification of clinical authorization.

    \item \textbf{Inadequate Support for Longitudinal Record Management with Privacy:}  
    Healthcare requires persistent, domain-specific identifiers (e.g., patient IDs) to link records across organizations~\cite{houtan2020survey}. However, exposing such identifiers directly creates linkability risks and profiling threats. 
    The tension between longitudinal record management and privacy is evident in the current state of the art of Health Information Exchange (HIE), which enables the sharing of patient records across different healthcare providers. 
    Javed et al.~\cite{javed2021health} proposed \textit{Health-ID}, a decentralized IDM framework for HIE on a consortium Ethereum blockchain. Patients undergo off-chain identity proofing by a healthcare administrator, who issues a signed JSON Web Token (JWT) containing selected attributes. The patient then deploys a Health Smart Contract (Health SC), whose blockchain address serves as their Health ID, and stores the hash of the JWT on the blockchain. A Registry Smart Contract, maintained by regulators, links the patient’s HealthID to their public key and validated attestations, enabling healthcare providers to authenticate patients consistently across domains.  
    Similarly, Zhuang et al.~\cite{zhuang2020patient} proposed a decentralized IDM for HIE on a permissioned Ethereum blockchain. In this design, the healthcare administrator creates a blockchain account for each patient, with the account address itself functioning as a global identifier linked to the patient’s Personal Health Information (PHI). Biometric data provided at registration is bound to the blockchain account and later used for authentication during follow-up visits, ensuring continuity across providers but reinforcing dependence on a persistent identifier.  
    More recent SSI-based frameworks, such as SSI-SHS by Bai et al.~\cite{bai2022self} and BDIMHS by Torongo et al.~\cite{torongo2023blockchain}, leverage Hyperledger Indy and Aries to support DIDs and Verifiable Credentials. These approaches mitigate linkability risks by enabling pairwise, pseudonymous interactions. However, they introduce a new challenge: without a privacy-preserving mechanism for linking health records to persistent identifiers, it becomes technically challenging for healthcare providers to maintain continuous longitudinal health histories, which are essential for safe and effective care. 
    Table~\ref{tab:health-id-comparison} summarizes this tension by comparing representative blockchain-based healthcare identity frameworks with respect to identifier exposure, longitudinal support, provider unlinkability, and conditional traceability.
    
    \item \textbf{Absence of Conditional Traceability:}  
    Healthcare demands accountability for both patients and providers to address fraud, malpractice, or administrative errors. Patients should be able to audit professional access to their records~\cite{shehu2024compliance}, and activities must be traceable under defined conditions rather than fully anonymous. Generic DIDM lacks native support for conditional traceability or patient-centered transparency.
\end{enumerate}

These limitations motivate the design of a HIDM framework that extends the SSI model with verifiable legitimacy, unlinkable pseudonyms for health records, and tamper-resistant, legally binding traceability.

\section{Design Overview of the HIDM Framework} 
\label{sec:design-overview-healthcare-idm}

This section presents the design of the proposed HIDM framework, which extends the Self-Sovereign Identity (SSI) model within the Decentralized Identity Management (DIDM) paradigm to address the unique requirements of healthcare identity management. The framework is built on three primary design principles, each formulated to overcome the corresponding limitation identified in Section~\ref{subsubsec:limitations-idm-healthcare}.

\begin{enumerate}

    \item \textbf{Verifiable legitimacy:}  
    Verifiable legitimacy is ensured by the Government Health Authority (GHA), which serves as the root of trust, issuing a digitally signed Legitimacy Verifiable Credential (L-VC) to each entity in the healthcare ecosystem. The L-VC attests that the holder is a legally recognized organization, ensuring that patients interact only with trusted and authorized entities~\cite{bai2022self}.

    \item \textbf{Unlinkable pseudonymity:} Unlinkable pseudonymity is provided by cryptographically generated pseudonyms linked to the healthcare-specific identifier of the patient using a PRE scheme (Section~\ref{subsubsec:issuance-pseudonym-token-key}), allowing the Health Record Repository (HRR) to recover the encrypted identifier while ensuring healthcare providers see only the pseudonym. Thus, cryptographic pseudonyms support consistent longitudinal record indexing while preventing linkability across different healthcare interactions~\cite{jensen2019pseudonymisation}.

    \item \textbf{Conditional traceability:}  
    Conditional traceability follows a separation-of-duties approach, in which identity-linkage responsibilities are divided between two independent authorities. As shown in Figure~\ref{fig:traceabilityblockdiagram}, resolving a pseudonymized event at a healthcare organization to a real-world identity requires a coordinated, multi-step resolution process. The Agency for PatientCare (APC) maintains the binding between a healthcare-specific $\mathsf{PatientID}$ and the patient’s verified real-world identity information ($\mathsf{PatientID} \rightarrow \mathsf{PII}$). In contrast, the Pseudonym Token Authority (PTA) maintains the mapping between $\mathsf{PatientID}$ and PTA-issued, event-scoped pseudonym tokens ($\mathsf{PatientID} \leftrightarrow \mathsf{PTI}$). Neither authority possesses sufficient information to independently link a pseudonymized event to a patient’s real-world identity. Traceability is therefore non-automatic and can be achieved only under explicitly authorized conditions when both authorities cooperate to reconstruct the identity chain~\cite{clark1987comparison,xu2024conditional}.

\end{enumerate}

\begin{figure}[h]
    \centering
    \includegraphics[width=8cm, angle=0, trim={0cm 0cm 0cm 0cm}, clip]{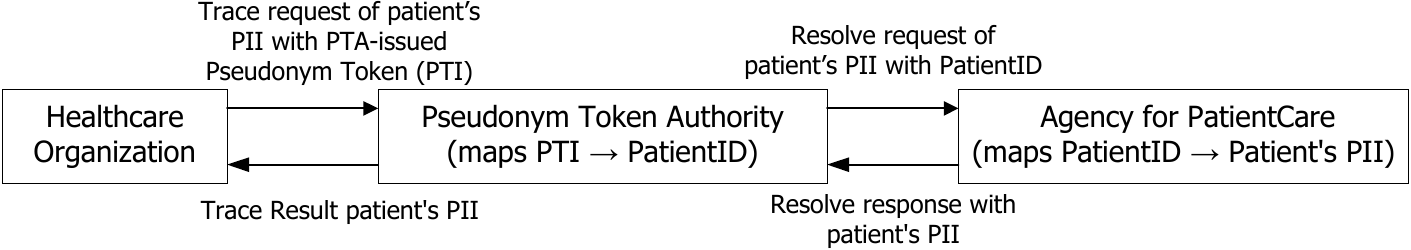}
    \caption{Coordinated identity reconstruction workflow showing the separation of duties between the PTA and APC}
    \label{fig:traceabilityblockdiagram}
\end{figure}

The remainder of this section is organized as follows: Section~\ref{subsec:architecture-healthcare-idm} describes the architectural components and their roles, and Section~\ref{subsec:patient-interaction-healthcare-focused-idm} presents the patient's end-to-end interaction with the IDM framework.

\subsection{Architecture of the HIDM} 
\label{subsec:architecture-healthcare-idm}

The HIDM framework comprises five core components: the Agency for PatientCare (APC), the Government Health Authority (GHA), the Healthcare Providers (including both organizations and individual professionals), the Health Record Repository (HRR), and the Auditor. Each component fulfills a distinct role in enabling a secure, efficient, and patient-centric healthcare ecosystem. Each component is required to have:

\begin{enumerate}
    \item a Decentralized Identifier (DID) demonstrating cryptographic control over its public keys.
    \item a Legitimacy Verifiable Credential (L-VC) attesting to its real-world legitimacy and operational authority.
\end{enumerate}

Patients likewise hold DIDs and demonstrate legitimacy through Patient Credentials issued by the APC (Section~\ref{subsubsec:patient_credential_issuance}).

\begin{figure*}[h]
 \centering
 \includegraphics[width=10cm]{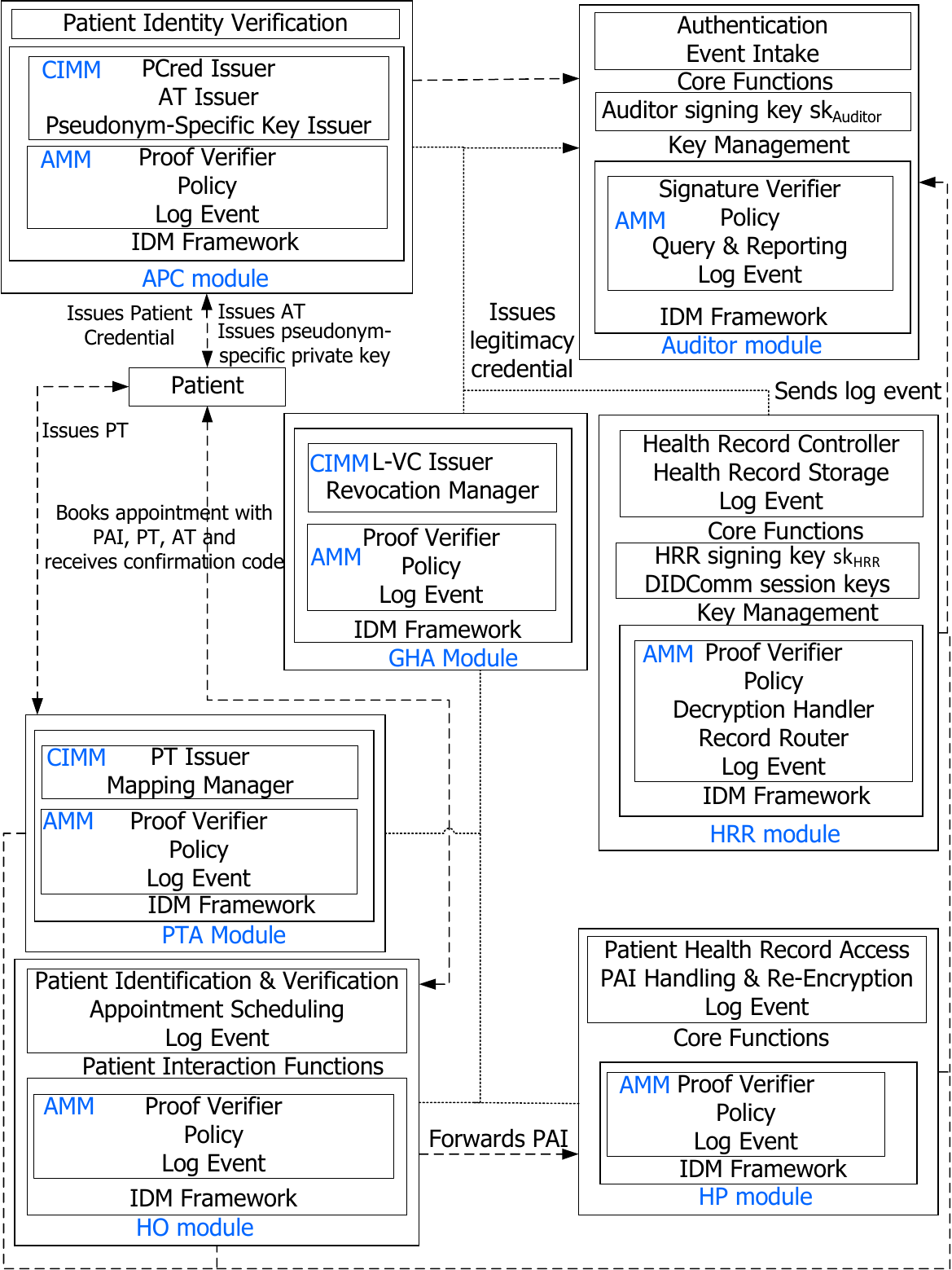}
 \caption{Architecture of the HIDM framework}
 \label{fig:block_diagram_2}
\end{figure*}

All components are implemented as software-based modules accessible to patients, healthcare professionals, and administrative staff. For example, patients interact with the APC, PTA, healthcare organization, HRR, and Auditor modules; healthcare professionals use the Healthcare Professional Module; and administrative staff use the Healthcare Organization Module for functions such as identity verification at check-in. Modules operate as automated services implemented as servers, APIs, or smart contracts without direct human intervention. To ensure interoperability with existing healthcare systems, the modules exchange data through secure, standardized APIs that align with established healthcare interoperability frameworks such as HL7 FHIR. Figure~\ref{fig:block_diagram_2} shows the block diagram of the HIDM framework, illustrating patient interactions with the APC, GHA, Healthcare Providers, HRR, and Auditor. Their individual functional roles are summarized below:

\begin{itemize}
    \item \textbf{GHA Module:} Serves as the root of trust, verifying the legitimacy and operational authority of all components. It issues digitally signed Legitimacy Verifiable Credentials $\mathsf{L\!-\!VC}$ to authorized entities (e.g., APC, HO, PTA).

    \item \textbf{APC Module:} Manages patient identity, issues Patient Credentials after verifying identity documents (Section~\ref{subsubsec:patient_credential_issuance}), generates pseudonym-specific private keys (Section~\ref{par:pseudonym-private-key}), and issues Appointment Tokens (Section~\ref{subsubsec:appointment-token-issuance}).

    \item \textbf{PTA Module:} Issues digitally signed Pseudonym Tokens that bind a patient's pseudonym to their healthcare-specific identifier (Section~\ref{par:pseudonym-token-issuance}).

    \item \textbf{Healthcare Provider Module:} Represents both healthcare organizations and healthcare professionals:
    \begin{itemize}
        \item \textbf{Healthcare Organization Module:} Represents clinics, hospitals, pharmacies, and diagnostic centers. Manages patient appointments and credential verification, enabling healthcare professionals to access authenticated health records.
        \item \textbf{Healthcare Professional Module:} Allows healthcare professionals to generate and update health records during patient interactions.
    \end{itemize}

    \item \textbf{HRR Module:} Maintains longitudinal health records with controlled access through two sub-modules:
    \begin{itemize}
        \item \textbf{Health Record Controller:} Processes and authorizes health record access requests. It exposes a RESTful API consistent with HL7 FHIR interoperability standards, enabling authorized healthcare professionals to query, retrieve, and update health data using standard FHIR resources (e.g., \texttt{Patient}, \texttt{Observation}, \texttt{DiagnosticReport}).
        
        \item \textbf{Health Record Storage:} Ensures secure health record storage and availability to authorized users.
    \end{itemize}

    \item \textbf{Auditor Module:} Ensures accountability and transparency by providing patients with access to their own interaction logs (\texttt{PatientAccessible} records) and granting the Auditor Authority full visibility into broader event records (\texttt{AuditorAuthorityAccessible}) for oversight, regulatory compliance, and dispute resolution (Section~\ref{subsubsec:immutable-ledgers-healthcare-idm}). When a request for log data is received, the Auditor Module verifies the requester’s identity and role through credential-based authentication: a patient authenticates using their Patient Credential ($\textsf{PCred}$), while institutional entities such as the APC, PTA, and Auditor Authority authenticate using their GHA-issued Legitimacy Verifiable Credential ($\textsf{L-VC}$). The module’s access-control logic, implemented via a smart contract, validates both the credential type and the associated role attributes before releasing data, thereby ensuring that the distinction between patient and authority is cryptographically verifiable and technically enforced.

\end{itemize}

These components rely on permissioned immutable ledgers, which ensure decentralization, verifiability, and auditability. The following section details the roles of these ledgers.

\subsubsection{Immutable Ledgers Supporting the HIDM Framework}
\label{subsubsec:immutable-ledgers-healthcare-idm}

The HIDM framework utilizes permissioned, immutable ledgers accessible via smart contracts to support Decentralized Identity Management (DIDM), enforce the single-use property of tokens, and provide accountability through transparent event logging. The ledgers are as follows:

\begin{itemize}
    \item \textbf{DID Document Ledger:} Stores each entity DID document to enable distributed public key discovery, a core function of the Decentralized Public Key Infrastructure (DPKI). Entities query the ledger to obtain up-to-date cryptographic data (e.g., public keys, service endpoints) for secure communications and credential verification.

    \item \textbf{Appointment Token Usage Ledger:} Records the unique identifier of every Appointment Token (Section~\ref{subsubsec:appointment-token-issuance}) used for appointment scheduling. This ensures tokens are single-use and prevents replay attacks.

    \item \textbf{Auditor Ledger:} Logs events generated by core modules in the healthcare ecosystem (e.g., Healthcare Organization, Healthcare Professional, APC, PTA, HRR). Each record follows a standardized structure:
    \begin{itemize}
        \item \textbf{logID:} Unique event identifier.
        \item \textbf{timestamp:} Date and time of the event.
        \item \textbf{originModule:} Module responsible for the event (e.g., APC, Healthcare Organization).
        \item \textbf{eventType:} Event descriptor (e.g., \texttt{PatientCredentialIssuance}, \texttt{HealthRecordRead}, \texttt{AppointmentBooked}).
        \item \textbf{accessLevel:} Access category (\texttt{PatientAccessible} or \texttt{AuditorAuthorityAccessible}).
        \item \textbf{patientIdentifier:} Healthcare-specific patient identifier or pseudonym.
        \item \textbf{healthcareProfessionalIdentifier:} Identifier for the healthcare professional, if applicable.
        \item \textbf{eventDetails:} Additional event data (e.g., ATI used, confirmation code, disclosed attributes).
    \end{itemize}
    
    Events are categorized by the \texttt{accessLevel} field as follows:
    \begin{itemize}
        \item \textbf{Patient-Accessible Event Records:} Provide patients with transparency into their healthcare interactions, such as:
        \begin{itemize}
            \item Appointment booking events.
            \item Identity verification events.
            \item Health record access (view or modification) by healthcare professionals.
        \end{itemize}
    
        \item \textbf{Auditor Authority-Accessible Event Records:} Provide the Auditor authority with information for oversight, compliance, and dispute resolution, including:
        \begin{itemize}
            \item Patient Credential issuance events.
            \item Pseudonym Token issuance events.
            \item Appointment Token issuance events.
            \item Health record access events by patients and healthcare professionals.
        \end{itemize}
    \end{itemize}
\end{itemize}

Table~\ref{tab:ledger-access} summarizes the scope of visibility across the different ledgers. In all cases, query metadata (e.g., requester identity, timestamp, verification outcome) is restricted to administrators and is not publicly visible.

\begin{table}[h]
\centering
\caption{Access Control Policies for Immutable Ledgers}
\label{tab:ledger-access}
\begin{tabular}{p{1.5cm} p{2cm} p{3cm}}\\
\hline
\textbf{Ledger} & \textbf{Authorized Writer} & \textbf{Authorized Reader} \\
\hline
\textbf{DID Document Ledger} & The DID Subject (e.g., Patient, HO, APC) for their own DID Document. & Public (all entities can resolve DID Documents). \\
\hline
\textbf{Appointment Token Usage Ledger} & The Healthcare Organization (HO), upon successful booking of an appointment. & Query-limited: Authorized entities (e.g., HO) may check the existence of a specific ATI but cannot browse the full ledger. \\
\hline
\textbf{Auditor Ledger} & Core system entities (APC, PTA, HO, HP, HRR) log their own actions. &
\begin{tabular}[t]{@{}l@{}}
\textbf{Tiered Access:} \\
1. Patients can access their own \\ \texttt{PatientAccessible} records. \\
2. The Auditor Authority can access all \\\texttt{AuditorAuthorityAccessible} \\ records 
 for oversight and compliance.
\end{tabular} \\
\hline
\end{tabular}
\end{table}

\subsection{Patient Interaction with HIDM Framework} \label{subsec:patient-interaction-healthcare-focused-idm}
This subsection presents the end-to-end patient journey within the HIDM framework, outlining the sequence of interactions required to access healthcare services. The process begins with the patient generating multiple sets of private-public pairs for distinct interactions (Section \ref{subsubsec:key-management-unlinkability}). Next, the patient acquires a Patient credential to prove the legitimacy of the patient within the healthcare ecosystem (Section \ref{subsubsec:patient_credential_issuance}). Then, the patient uses the Patient Credential to acquire a Pseudonym Token and a Pseudonym-Specific Private Key to enable unlinkable authentication (Section \ref{subsubsec:issuance-pseudonym-token-key}). The patient further acquires a single-use Appointment Token to reserve an appointment at the healthcare organization (Section \ref{subsubsec:appointment-token-issuance}). Finally, the patient utilizes these tokens to access healthcare services (Section \ref{subsubsec:patient_access_healthcare_provider}). 

\subsubsection{Patient Key Management}
\label{subsubsec:key-management-unlinkability}

The HIDM's strategy to manage patient keys is designed to provide robust unlinkability by default while empowering patients with control over the generation of keys. The strategy involves the patient managing several categories of cryptographic key pairs within a secure digital wallet. Each category is bound to a distinct function, and separation ensures that keys are not reused across unrelated contexts, thereby preventing linkability. 

\begin{enumerate}

  \item \textbf{Patient Identity Key Pair} 
  $(\mathsf{sk}_{\text{patient-L}}, \mathsf{pk}_{\text{patient-L}})$: \\
  \textbf{Duration:} Long-lived (persistent across the entire healthcare ecosystem). \\
  \textbf{Usage:} Anchors the patient’s cryptographic identity. The public key $\mathsf{pk}_{\text{patient-L}}$ is embedded in the $\mathsf{PatientCredential}$ to prove ownership by the patient (see Section~\ref{subsubsec:patient_credential_issuance}). 

  \item \textbf{Pseudonym Generation Key Pair} 
  $(\mathsf{sk}_{\mathsf{Pgen}}, \mathsf{pk}_{\mathsf{Pgen}})$: \\
  \textbf{Duration:} Healthcare-organization-specific and renewable. A fresh key pair can be generated for each healthcare organization ($\mathsf{HO}_i$). \\
  \textbf{Usage:} Generates the patient pseudonym $P_{\mathsf{patient}}$ and the re-encryption key 
  $rk_{\mathsf{patient} \rightarrow \mathsf{HRR}}$ for the Health Record Repository ($\mathsf{DID}_{\mathsf{HRR}}$) (see Section~\ref{par:generation-pseudonym}). The public key $\mathsf{pk}_{\mathsf{Pgen}}$ is embedded in the pseudonym.

  \item \textbf{Pseudonym-Specific Signing Key} 
  $(\mathsf{sk}_{P_{\mathsf{patient}}})$: \\
  \textbf{Duration:} Healthcare-organization-specific and renewable. \\
  \textbf{Usage:} Generates unlinkable digital signatures under a pseudonym (see Section~\ref{par:pseudonym-private-key}). These signatures authenticate requests such as $\mathsf{AppointmentScheduleRequest}$ and $\mathsf{IdentityVerificationRequest}$ (see Section~\ref{subsubsec:patient_access_healthcare_provider}).

  \item \textbf{DIDComm Communication Keys}: \\
  \textbf{Duration:} Ephemeral (rotated per session). \\
  \textbf{Usage:} Generated independently of all other keys, and the corresponding public keys are included in the patient’s DID Document. Used exclusively for establishing mutually authenticated DIDComm channels with the APC ($\mathsf{DID}_{\mathsf{APC}}$), the PTA ($\mathsf{DID}_{\mathsf{PTA}}$), and healthcare providers ($\mathsf{DID}_{\mathsf{HO}_i}$). 

\end{enumerate}

Patients may choose to reuse the same pseudonym and signing key within a healthcare organization to benefit from continuity features such as priority booking or loyalty services. Alternatively, they may generate new pseudonyms and keys for repeated visits, even within the same organization, if stronger unlinkability across interactions is desired.

When accessing services from a healthcare organization, the patient employs these keys in sequence: the long-lived identity key pair anchors the $\mathsf{PatientCredential}$, a fresh pseudonym generation key pair and pseudonym-specific signing key are scoped to the target organization (may be rotated across visits), and a distinct DIDComm key pair is generated for each secure channel established with the APC, PTA, and the healthcare organization itself.

\subsubsection{Issuance of Patient Credential}
\label{subsubsec:patient_credential_issuance}

The Patient Credential is a verifiable and trusted digital identity for the patient, issued by the Agency for PatientCare (APC). It enables the patient to access healthcare services across multiple providers (e.g., health clinics) through the HIDM framework. The steps for Patient Credential issuance are described below:

\begin{itemize}

    \item \textbf{Legitimacy Verification:} 
    The APC demonstrates its legitimacy to the patient by presenting its GHA-issued $\mathsf{L\!-\!VC}_{\mathrm{APC}}$. This credential is digitally signed by the GHA; the patient verifies the signature using the GHA’s public key (pre-distributed or published on the DID Document Ledger). 
    
    \item \textbf{Submission of Personal Information:} 
    The patient submits their real-world Personally Identifiable Information (PII) (e.g., name, address, biometric data, and a government-issued ID) to the APC.

    \item \textbf{Verification by APC:} 
    The APC verifies the received PII and cross-checks the submitted details for consistency.
    
    \item \textbf{Credential Generation and Signature:} 
    Upon successful verification, the APC generates the Patient Credential and digitally signs it using the CL signature scheme. The format of the Patient Credential is defined as:
    $$
    \begin{aligned}
    \mathsf{PCred} =\ &
    \big[ \big( 
        \mathsf{CredentialID},\ 
        \mathsf{DID}_{\mathsf{Patient\_L}},\ 
        \mathsf{PatientID}, \\
    &\quad
        \mathsf{IssueDate},\ 
        \mathsf{BioHash},\ 
        \mathsf{DID}_{\mathsf{APC}}
    \big) \big],\ 
    \mathsf{Sig}_{\mathsf{APC}}
    \end{aligned}
    $$

    where:
    \begin{itemize}
        \item $ \mathsf{CredentialID} $: Unique identifier of the  Patient Credential 
        \item $ \mathsf{DID}_{\mathsf{Patient\_L}} $: Long-lived DID generated by the patient and bound to the PatientCredential
        \item $ \mathsf{PatientID} $: Healthcare-specific patient identifier
        \item $ \mathsf{IssueDate} $: Date of issuance of the credential
        \item $ \mathsf{BioHash} $: Probabilistic hash of the patient's biometric data (e.g., facial features) (Section~\ref{subsubsec:identity-credential-role-idm})
        \item $ \mathsf{DID}_{\mathsf{APC}} $: DID of the APC
        \item $ \mathsf{Sig}_{\mathsf{APC}} $: Digital signature over all listed fields, generated by the APC using the CL signature scheme
    \end{itemize}

    \item \textbf{Event Logging and Auditability:} 
    The APC maintains a secure database~\cite{Mousa2020} to store the association between the patient’s PII and the PatientID of the issued Patient Credential. It also records the credential issuance event in the Auditor Ledger using $ \mathsf{PatientID} $ as $ \mathsf{patientIdentifier} $, with relevant $ \mathsf{eventDetails} $ such as $ \mathsf{CredentialID} $ and issuance timestamp (Section~\ref{subsec:architecture-healthcare-idm}).

\end{itemize}

Figure~\ref{fig:patient_credential_issuance} shows the sequence diagram of Patient Credential issuance. It presents how the patient provides identity attributes and proofs while the APC verifies and signs the credential.

\begin{figure}[h]
    \centering
    \includegraphics[width=8cm, angle=0, trim={0cm 0cm 0cm 0cm}, clip]{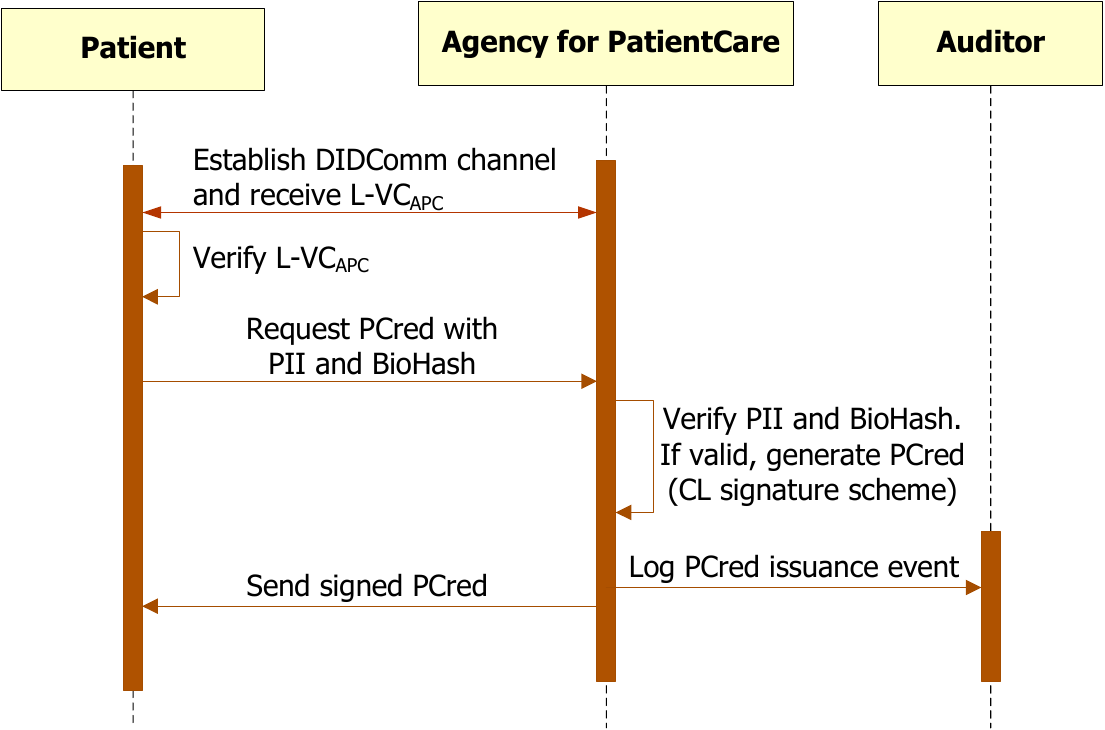}
    \caption{Sequence diagram illustrating the issuance of a Patient Credential}
    \label{fig:patient_credential_issuance}
\end{figure}

\subsubsection{Issuance of Pseudonym Token and Pseudonym-Specific Private Key} \label{subsubsec:issuance-pseudonym-token-key}

The pseudonym is a fundamental element of the HIDM framework. It enables unlinkable interactions with healthcare providers while supporting consistent indexing of patient records in the HRR. For each appointment booking, the patient derives a fresh visit-scoped pseudonym $\mathsf{P}_{\mathsf{patient}}$, ensuring that repeated visits to the same healthcare organization cannot be linked by default. The issuance process comprises three steps, each detailed in the following subsections: the generation of a patient pseudonym that is cryptographically linked to the healthcare-specific patient identifier in Section~\ref{par:generation-pseudonym}, the issuance of a digitally signed Pseudonym Token by the PTA, certifying the binding between the pseudonym and an authenticated patient in Section~\ref{par:pseudonym-token-issuance}, and the issuance of a pseudonym-specific private key to the patient, which enables traceable yet unlinkable authentication within the healthcare ecosystem in Section~\ref{par:pseudonym-private-key}.

\paragraph{Generation of the Pseudonym}
\label{par:generation-pseudonym}

A key challenge in the HIDM framework is enabling consistent indexing of patient records in the HRR using a persistent healthcare-specific identifier while preventing its disclosure to healthcare providers. Naïve approaches, such as hashing or deterministic pseudonyms, fail because they create globally persistent identifiers visible to all parties, enabling cross-institution correlation and long-term tracking.

To address this challenge, the framework employs a Proxy Re-Encryption (PRE)~\cite{mambo1997proxy} scheme, applied within the HIDM context to enable domain transformation without revealing the underlying identifier. In essence, the patient encrypts their healthcare-specific identifier once and delegates controlled re-encryption to the HRR. The design ensures that the HRR can maintain consistent longitudinal records, while healthcare organizations and other intermediaries remain unable to link records across domains.

The patient generates a pseudonym $ \mathsf{P}_{\mathsf{patient}} $ encapsulating their encrypted identifier, together with a re-encryption key $ \mathsf{rk}_{\mathsf{patient} \rightarrow \mathsf{HRR}} $. The re-encryption key allows a healthcare organization to transform $ \mathsf{P}_{\mathsf{patient}} $ into an HRR-specific pseudonym $ \mathsf{P}_{\mathsf{HRR}} $ without learning the identifier itself. The HRR, using its private key $ \mathsf{sk}_{\mathsf{HRR}} $, can then recover the encrypted identifier to index longitudinal records but never the patient’s real-world identity. Control over pseudonym generation and re-encryption delegation resides entirely with the patient, preserving autonomy and unlinkability across healthcare providers. To reinforce unlinkability, the patient may generate new pseudonyms for different healthcare organizations or rotate them for repeated visits to the same organization. This prevents cross-provider correlation while still allowing the HRR to maintain consistent and complete longitudinal records. Full technical details are provided below.

\paragraph*{Proxy Re-Encryption Assumptions}
The HIDM framework relies on a unidirectional, single-hop, non-interactive PRE scheme based on bilinear pairings. Specifically, the design assumes a pairing-based PRE family modeled after the construction by Ateniese et al. ~\cite{ateniese2006improved}, which achieves indistinguishability under chosen-plaintext attacks (IND-CPA). The PRE mechanism is employed under the following cryptographic assumptions:

\begin{itemize}
\item \textbf{Unidirectionality:} The re-encryption key $\mathsf{rk}_{\mathsf{patient} \rightarrow \mathsf{HRR}}$ enables a one-way transformation from the patient’s domain to the HRR domain. This prevents the HRR from deriving the patient’s secret key or performing inverse transformations beyond its authorized scope.
\item \textbf{Single-Hop:} To prevent unauthorized re-encryption chaining and unintended linkability, a ciphertext re-encrypted for the HRR cannot be further re-encrypted for any other entity.
\item \textbf{Non-Interactivity:} The patient can generate the re-encryption key independently using the HRR’s public key $\mathsf{pk}_{\mathsf{HRR}}$, ensuring user autonomy without requiring an interactive key exchange with the HRR.
\end{itemize}

\paragraph*{Integration of Proxy Re-Encryption for Pseudonym Management}
\label{subpar:healthcare_proxy_reencryption}

The integration of the PRE scheme into the proposed HIDM framework follows its standard construction, which consists of the following steps: 

\begin{itemize}

    \item \textbf{Public Parameters:}  
    The public parameters for the scheme are defined within a bilinear pairing setting \cite{boneh2001identity}. Let $ E / \mathbb{F}_q $ be an elliptic curve defined over a finite field $ \mathbb{F}_q $ of order $ q $. Let $ G_1 \subseteq E(\mathbb{F}_q) $ be an additive cyclic subgroup of order $ r $, and $ G_2 \subseteq E'(\mathbb{F}_q) $ be another additive cyclic subgroup of the same order $ r $, defined either on a distinct curve $ E' $ or on a separate subgroup to ensure non-isomorphism with $ G_1 $. Let $ g_1 $ and $ g_2 $ be generators of $ G_1 $ and $ G_2 $. Let $ G_T $ be a multiplicative cyclic subgroup of order $ r $ contained in the finite-field extension $ \mathbb{F}_{q^k}^* $, where $ k $ is the embedding degree satisfying $ r \mid (q^k - 1) $. This construction is necessary because the bilinear pairing $ e : G_1 \times G_2 \rightarrow G_T $ requires a target group whose order divides $ q^k - 1 $; consequently, $ G_T $ resides in the larger field $ \mathbb{F}_{q^k}^* $. The scalar field $ \mathbb{Z}_r $ consists of integers modulo $ r $ from which private keys, exponents, and random values are selected. The pairing result $ z = e(g_1, g_2) \in G_T $.

    \item \textbf{Key Generation:} \\
    Patient (pseudonym generation): $\mathsf{sk}_{\mathsf{patient}} = x \in \mathbb{F}_r$, $\mathsf{pk}_{\mathsf{patient}} = x g_2 \in G_2$. \\[4pt]
    HRR: $\mathsf{sk}_{\mathsf{HRR}} = y \in \mathbb{F}_r$, $\mathsf{pk}_{\mathsf{HRR}} = y g_1 \in G_1$.

    \item \textbf{Pseudonym Generation:}  
     The steps of the pseudonym generation process are provided below:  
     
    \begin{itemize}
    
        \item Map the healthcare-specific identifier $ \mathsf{PatientID} $ to $ \mathbb{F}_r $ using a hash-to-field:
        $$
            h = \mathsf{HashToField}(\mathsf{PatientID}) \in \mathbb{F}_r.
        $$
        
        \item Encode it in $G_T$ via the pairing result:
        $$
            \mathsf{PatientID}_{G_T} = z^{h} \in G_T.
        $$
        
        \item A symmetric encryption key $k$ is derived using the HMAC-based Extract-and-Expand Key Derivation Function (HKDF)~\cite{rfc5869}. Specifically, a fixed 32-byte system-wide salt $\mathsf{salt}$ and a context-specific information string \texttt{HIDM-\allowbreak PRE-\allowbreak PatientID-\allowbreak Derivation-\allowbreak v1} are used:
        $$
        k = \mathsf{HKDF}\big(
        \mathsf{salt},
        \mathsf{enc}(\mathsf{PatientID}_{G_T}),
        \mathsf{info}
        \big)
        $$
        
        \item Encrypt the original identifier under $k$ using an authenticated encryption scheme (e.g., AES-GCM~\cite{NIST800-38D}):
        $$
            \mathsf{ct}_{\mathsf{PatientID}} = \mathsf{Enc}(\mathsf{PatientID}, k).
        $$
        
        \item Select a random $ r \in \mathbb{F}_r $ and compute the pseudonym components:
        $$
        \begin{aligned}
            \mathsf{P}_{\mathsf{patient},1} 
                &= z^r \cdot \mathsf{PatientID}_{G_T} \in G_T, \\
            \mathsf{P}_{\mathsf{patient},2} 
                &= r\,\mathsf{pk}_{\mathsf{patient}}
                 = r(x g_2) \in G_2.
        \end{aligned}
        $$

        \item Form the pseudonym:
        $$
            \mathsf{P}_{\mathsf{patient}} = (\mathsf{P}_{\mathsf{patient},1},\, \mathsf{P}_{\mathsf{patient},2}).
        $$
        
        \item Compute the re-encryption key for the HRR:
        $$
        	\mathsf{rk}_{\mathsf{patient} \rightarrow \mathsf{HRR}} = \tfrac{1}{x}\,\mathsf{pk}_{\mathsf{HRR}} 
        	= \tfrac{1}{x}(y g_1) = \tfrac{y}{x}\,g_1 \in G_1.
        $$

        \item Send the Pseudonym Access Information (PAI) to the healthcare organization:
        $$
            \mathsf{PAI} = \big( \mathsf{P}_{\mathsf{patient}},\ \mathsf{rk}_{\mathsf{patient} \rightarrow \mathsf{HRR}},\ \mathsf{ct}_{\mathsf{PatientID}} \big).
        $$
    \end{itemize}

    \item \textbf{Healthcare Organization Transformation:} 
    Verify the re-encryption key via:
    $$
        e(\mathsf{rk}_{\mathsf{patient} \rightarrow \mathsf{HRR}}, \mathsf{pk}_{\mathsf{patient}}) \stackrel{?}{=} e(\mathsf{pk}_{\mathsf{HRR}}, g_2).
    $$
    Upon success, transform the patient pseudonym into an HRR-specific pseudonym:
    $$
    \begin{aligned}
    \mathsf{P}_{\mathsf{HRR}} &=
    \big( \mathsf{P}_{\mathsf{patient},1},\, e(\mathsf{rk}_{\mathsf{patient} \rightarrow \mathsf{HRR}},\, \mathsf{P}_{\mathsf{patient},2}) \big) \\
    &= \big( z^r \cdot \mathsf{PatientID}_{G_T},\, z^{r y} \big).
    \end{aligned}
    $$

    \item \textbf{HRR Recovery:}  
    Recover $ \mathsf{PatientID}_{G_T} $ as:
    $$
        \frac{\mathsf{P}_{\mathsf{HRR},1}}{\big(\mathsf{P}_{\mathsf{HRR},2}\big)^{1 / y}} = \mathsf{PatientID}_{G_T},
    $$
    derive:
    $$
    k = \mathsf{HKDF}\big(
    \mathsf{salt},
    \mathsf{enc}(\mathsf{PatientID}_{G_T}),
    \mathsf{info}
    \big),
    $$
    and decrypt:
    $$
        \mathsf{PatientID} = \mathsf{Dec}(\mathsf{ct}_{\mathsf{PatientID}}, k).
    $$
    
\end{itemize}

\paragraph{Issuance of Pseudonym Token}
\label{par:pseudonym-token-issuance}

In the proposed HIDM framework, patients must prove their legitimacy to healthcare organizations without exposing their persistent healthcare-specific identifier $\mathsf{PatientID}$ from the Patient Credential. Direct disclosure of $\mathsf{PatientID}$ would enable linkability across visits, thereby undermining patient privacy.

To mitigate this risk, the HIDM introduces the Pseudonym Token, a credential issued by the Pseudonym Token Authority (PTA) that certifies the binding between a patient-generated pseudonym and the $\mathsf{PatientID}$ of the Patient Credential. It enables patients to authenticate under their pseudonyms, preventing identifier-based linkability while still ensuring verifiable legitimacy.

Although the patient generates the pseudonym, it cannot be arbitrarily chosen. To ensure correctness, the patient submits a Pseudonym Binding Proof (PBP) to the PTA, a NIZK proof that demonstrates the pseudonym is derived from their $\mathsf{PatientID}$ without revealing the derivation itself. Upon successful verification, the PTA issues the pseudonym token using the Schnorr signature scheme. The Pseudonym Token generation process consists of the following steps:

\begin{itemize}

    \item \textbf{Legitimacy Verification:}  
    The PTA presents its GHA-issued $\mathsf{L\!-\!VC}_{\mathrm{PTA}}$ for the patient to verify its operational legitimacy after establishing a mutually authenticated DIDComm channel.

    \item \textbf{Submission of Proofs:}  
    The patient submits:
    \begin{itemize}
        \item Healthcare-specific patient identifier $ \mathsf{PatientID} $
        \item Proof of knowledge of the Patient Credential $ \mathsf{PoK}_{\mathsf{PCred}} $
        \item Pseudonym $ \mathsf{P}_{\mathsf{patient}} $
        \item Pseudonym Binding Proof (PBP), a NIZK proof binding $ \mathsf{P}_{\mathsf{patient}} $ to $ \mathsf{PatientID} $ (Appendix~\ref{appendix:pseudonym_binding_proof})
    \end{itemize}

    \item \textbf{Verification by PTA:}
    The PTA verifies both the $ \mathsf{PoK}_{\mathsf{PCred}} $ and the PBP through its AMM.

    \item \textbf{Token Generation and Signature:} 
    Upon successful verification, the PTA initiates the Pseudonym Token generation process. It generates a unique Pseudonym Token Identifier $ \mathsf{PTI} $ and concatenates it with the patient pseudonym $ \mathsf{P}_{\mathsf{patient}} $. It then computes the challenge $ \mathsf{c} $ using the concatenated message and the commitment of a nonce $(\mathsf{r}, \mathsf{R})$, followed by the computation of the response $ \mathsf{s} $ from the nonce $\mathsf{r}$ to finalize the signature $(\mathsf{c}, \mathsf{s})$.

    \item \textbf{Event Logging and Auditability:}  
    The PTA records the pseudonym–PatientID association in a secure database~\cite{Mousa2020} and logs the token issuance event in the Auditor Ledger using $ \mathsf{P}_{\mathsf{patient}} $ as $ \mathsf{patientIdentifier} $, with relevant $ \mathsf{eventDetails} $ such as $ \mathsf{PTI} $ and issuance timestamp (Section~\ref{subsec:architecture-healthcare-idm}).
    
\end{itemize}

\begin{figure}[h]
    \centering
    \includegraphics[width=8cm, trim={0cm 0cm 0cm 0cm}, clip]{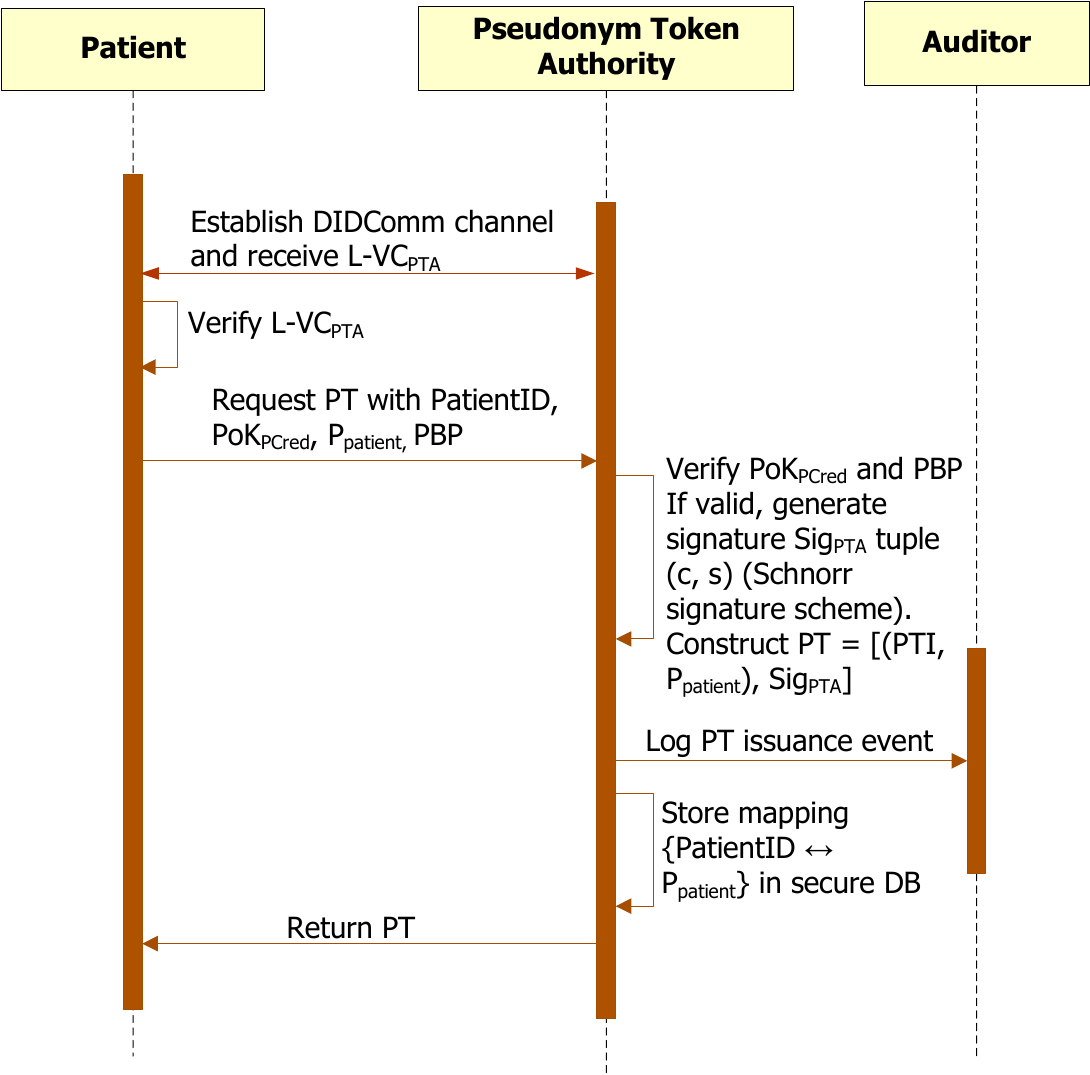}
    \caption{Sequence diagram of Pseudonym Token issuance}
    \label{fig:pseudonym_token_issuance}
\end{figure}

The Pseudonym Token format is:
$$
\mathsf{PT} = \left[ \big(\mathsf{PTI},\ \mathsf{P}_{\mathsf{patient}}\big),\ \mathsf{Sig}_{\mathsf{PTA}} \right]
$$
where:
\begin{itemize}
    \item $ \mathsf{PTI} $: Unique identifier of the Pseudonym Token
    \item $ \mathsf{P}_{\mathsf{patient}} $: Patient pseudonym
    \item $ \mathsf{Sig}_{\mathsf{PTA}} $: The signature tuple  $(\mathsf{c}, \mathsf{s})$ generated based on the Schnorr signature scheme by the PTA
\end{itemize}

Figure~\ref{fig:pseudonym_token_issuance} shows the sequence diagram of Pseudonym Token issuance. It presents how the patient submits their pseudonym together with proofs of credential possession and correct binding, while the PTA verifies the proofs and signs the pseudonym token.

\paragraph{Issuance of Pseudonym-Specific Private Key}
\label{par:pseudonym-private-key}

The pseudonym-specific private key $\mathsf{sk}_{P_{\mathsf{patient}}}$ enables the patient to authenticate actions under their pseudonym by generating verifiable digital signatures that ensure message integrity and non-repudiation, without exposing their underlying identity. To prevent the issuing authority APC from learning the pseudonym itself, the Blind IBS signature scheme is employed, allowing the APC to issue the private key without knowledge of the underlying actual pseudonym. The pseudonym-specific private key generation process consists of the following steps:

\begin{itemize}

    \item \textbf{Legitimacy Verification:}  
    The patient may optionally verify the APC's GHA-issued $\mathsf{L\!-\!VC}_{\mathrm{APC}}$ to confirm that the APC's DID has not been revoked before establishing a mutually authenticated DIDComm channel.

    \item \textbf{Blinded Pseudonym Submission:}  
    The patient blinds their pseudonym to obtain $P_{\mathsf{patient}}^{\prime}$ and sends it to the APC along with:
    \begin{itemize}
        \item Healthcare-specific patient identifier $ \mathsf{PatientID} $
        \item Proof of knowledge of the Patient Credential $ \mathsf{PoK}_{\mathsf{PCred}} $
    \end{itemize}

    \item \textbf{Proof Verification and Blinded Key Generation:}  
    The APC verifies the submitted proof using the AMM of its IDM framework and generates the blinded private key $\mathsf{sk}_{P_{\mathsf{patient}}^{\prime}}$ corresponding to $P_{\mathsf{patient}}^{\prime}$. The blinded key is returned to the patient.

    \item \textbf{Unblinding:}  
    The patient unblinds $\mathsf{sk}_{P_{\mathsf{patient}}^{\prime}}$ to obtain the true pseudonym-specific private key $\mathsf{sk}_{P_{\mathsf{patient}}}$.

    \item \textbf{Event Logging and Auditability:}  
    The APC records the issuance event in the Auditor Ledger using $ \mathsf{P}_{\mathsf{patient}} $ as $ \mathsf{patientIdentifier} $, with relevant $ \mathsf{eventDetails} $ such as issuance timestamp (Section~\ref{subsec:architecture-healthcare-idm}).

\end{itemize}

\begin{figure}[h]
    \centering
    \includegraphics[width=8cm, trim={0cm 0cm 0cm 0cm}, clip]{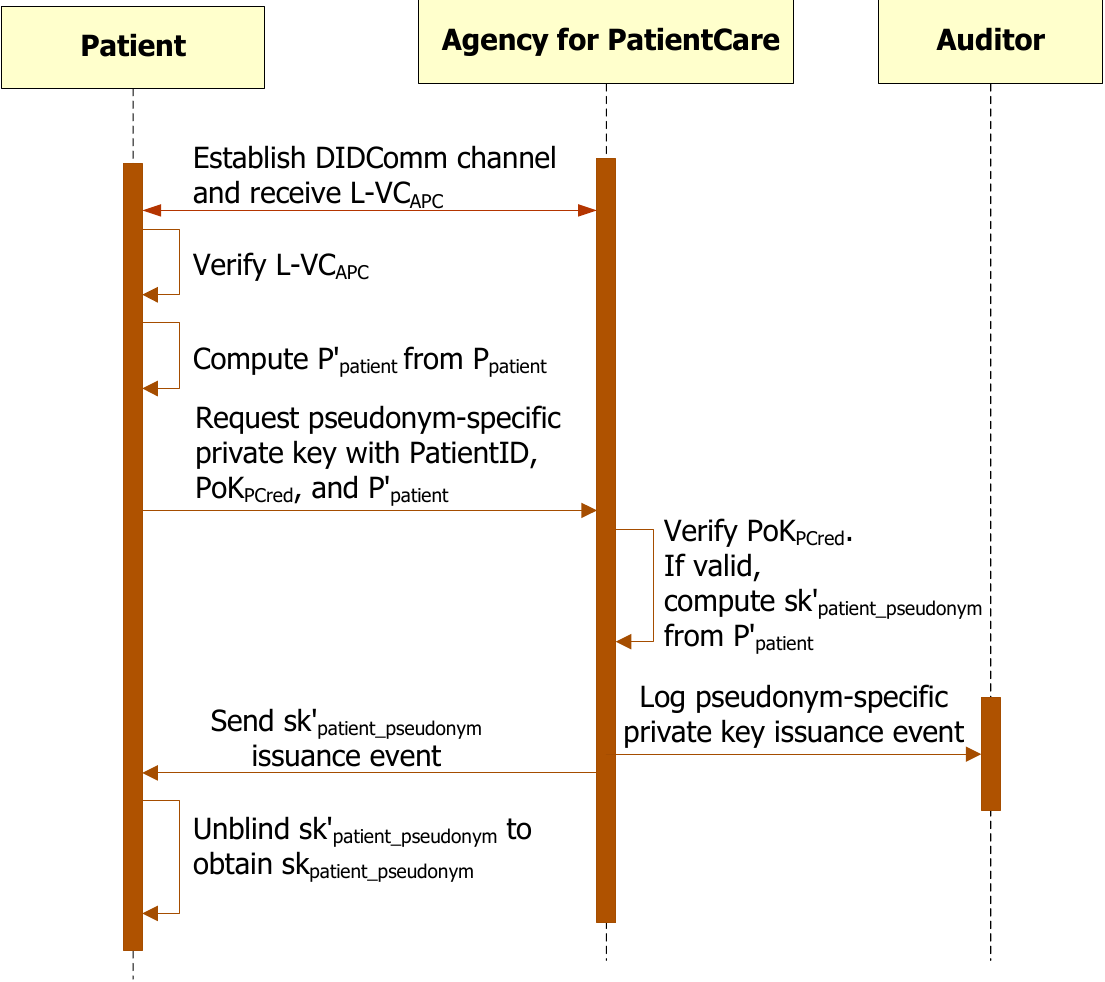}
    \caption{Sequence diagram of pseudonym-specific private key issuance}
    \label{fig:pseudonym_privatekey_issuance}
\end{figure}

\subsubsection{Issuance of Appointment Token}
\label{subsubsec:appointment-token-issuance}

The Appointment Token is generated by the APC using the Partially Blind Schnorr Signature scheme, which allows the patient to keep their self-generated Appointment Token Identifier (ATI) hidden from the APC while simultaneously enabling the APC to embed the expiration date-time into a single verifiable signature.

In this scheme, the patient’s self-generated Appointment Token Identifier (ATI) remains hidden from the APC, ensuring that the APC cannot later correlate appointments through the ATI. At the same time, the APC embeds the expiration date-time into a single verifiable signature, preventing token misuse beyond the permitted window.

\begin{itemize}

    \item \textbf{Legitimacy Verification:} 
    The patient may optionally re-verify the $\mathsf{L\!-\!VC}_{\mathrm{APC}}$ to ensure the APC's DID has not been revoked before establishing a mutually authenticated DIDComm channel.

    \item \textbf{Signature Protocol Initiation:}  the patient sends a signature request. The APC responds by generating a nonce $r$ and sending its commitment $R = g^r \pmod{p}$ to the patient.

    \item \textbf{ATI Generation and Commitment:} 
    The patient generates an Appointment Token Identifier (ATI) using the UUIDv4 function~\cite{davis2024rfc}, which produces a 128-bit pseudorandom identifier with negligible collision probability. 
    
    The patient computes the ATI commitment $\mathsf{commit}_{\mathsf{ATI}}$ using a nonce $t$. Next, the patient computes its challenge $cu_{\mathsf{patient}}^{\prime}$ using the blinded commitment of the APC's nonce $R^{\prime}$ and the ATI commitment $\mathsf{commit}_{\mathsf{ATI}}$. Then, the patient computes the blinded challenge $cu_{\mathsf{patient}}$.

    \item \textbf{Submission to APC:} 
    The patient submits the blinded challenge $cu_{\text{patient}}$, the $\mathsf{PatientID}$, and a proof of knowledge of the Patient Credential $\mathsf{PoK}_{\mathsf{PCred}}$ to the APC.

    \item \textbf{Verification by APC:} 
    The APC verifies $\mathsf{PoK}_{\mathsf{PCred}}$ using the AMM of its IDM framework.

    \item \textbf{Token Generation and Signature:} 
    Upon successful verification, the APC sets the expiration date-time, $\mathsf{exp}$, for the Appointment Token. The APC computes its challenge $cs_{\mathsf{APC}}$ using its public key $\mathsf{Y}_{\mathsf{APC}}$ and $\mathsf{exp}$.
    
    Then, it computes an intermediate combined challenge $c' = cu_{\mathsf{patient}} + cs_{\mathsf{APC}}$. Next, it computes the blinded response $s^{\prime} = r + c' \cdot x \pmod{q}$. It sends $s^{\prime}$ and $\mathsf{exp}$ to the patient.
    
    The patient unblinds $s^{\prime}$ to get $s$, and reconstructs the APC's challenge $cs_{\mathsf{APC}}$. Finally, the patient computes the final challenge $c = cu_{\mathsf{patient}}^{\prime} + cs_{\mathsf{APC}}$ to finalize the signature $(s, c, t)$.

    \item \textbf{Token Structure:} 
    The format of the Appointment Token is:
    $$
    \mathsf{AT} = \left[ \big(\mathsf{ATI},\ \mathsf{exp}\big),\ \mathsf{Sig}_{\mathsf{APC}} \right]
    $$
    where:
    \begin{itemize}
        \item $\mathsf{ATI}$: Appointment Token Identifier
        \item $\mathsf{exp}$: Expiration date-time
        \item $\mathsf{Sig}_{\mathsf{APC}}$: The signature tuple (c, s, t) generated based on the Partially blind schnorr signature scheme by the APC
    \end{itemize}

    \item \textbf{Event Logging and Auditability:} 
    The APC records the Appointment Token issuance event in the Auditor Ledger using $ \mathsf{P}_{\mathsf{patient}} $ as $ \mathsf{patientIdentifier} $, with relevant $ \mathsf{eventDetails} $ such as $ \mathsf{ATI} $, expiration time, and issuance timestamp (Section~\ref{subsec:architecture-healthcare-idm}).

\end{itemize}

Figure~\ref{fig:appointment_token_issuance} shows the sequence diagram of Appointment Token issuance. It presents how the patient binds the ATI and proves credential possession, while the APC embeds the expiration date and signs the one-time appointment token.

\begin{figure}[h]
    \centering
    \includegraphics[width=8cm]{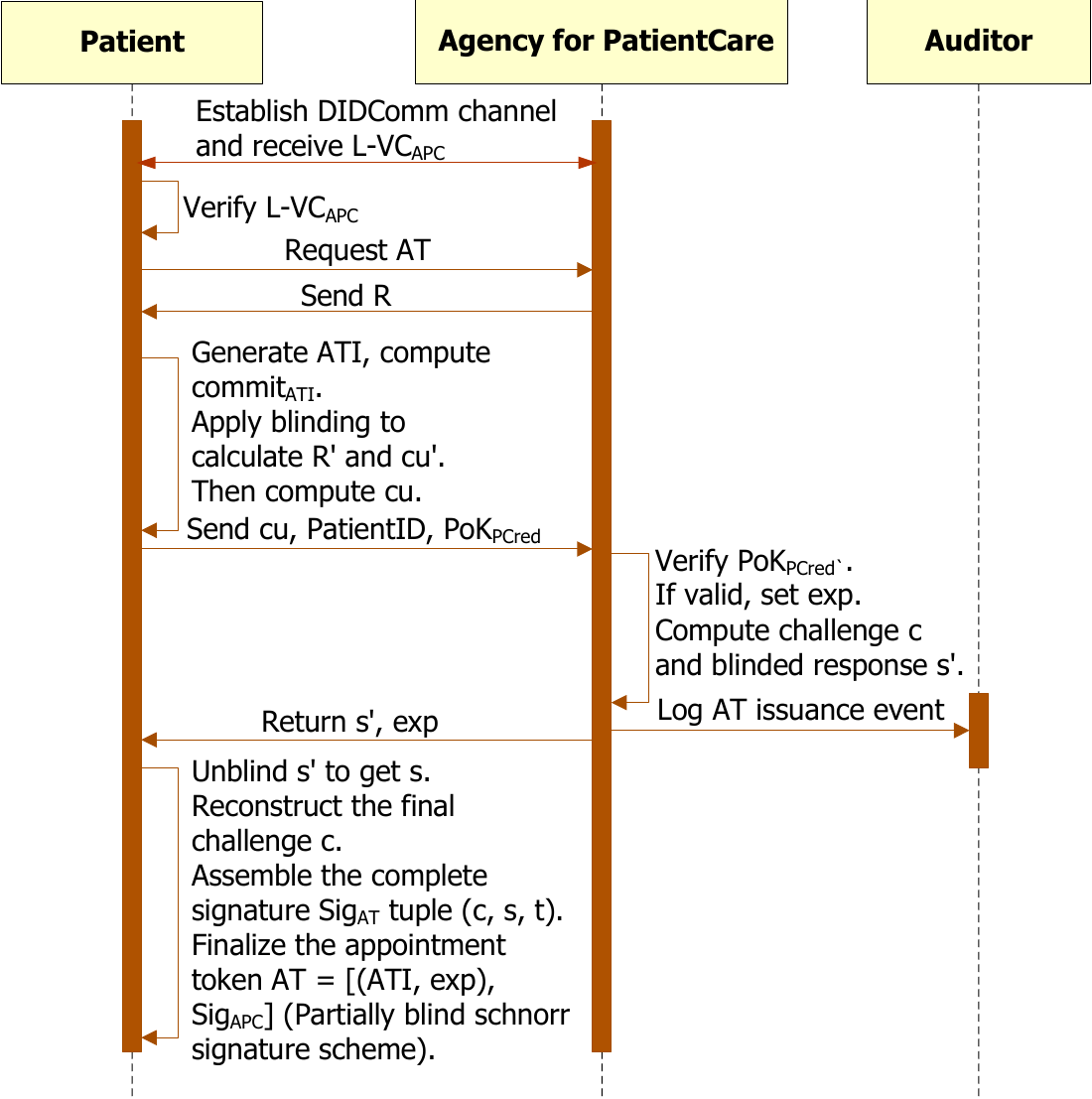}
    \caption{Sequence diagram of Appointment Token issuance}
    \label{fig:appointment_token_issuance}
\end{figure}

\subsubsection{Patient Access with a Healthcare Provider}
\label{subsubsec:patient_access_healthcare_provider}

Patient access within the HIDM framework occurs in two phases: interaction with the healthcare organization for appointment scheduling and in-person identity verification (Section \ref{patient_interaction_with_healthcare_organization}), followed by direct engagement with the healthcare professional for medical consultation (Section \ref{patient_interaction_with_healthcare_professional}).

\paragraph{Patient Interaction with Healthcare Organization}
\label{patient_interaction_with_healthcare_organization}

Patient interaction with a healthcare organization involves two stages: First, the patient remotely books an appointment using their single-use Appointment Token to secure a time slot (Section \ref{appointment_booking_pseudonym_appointment_booking}). Then, upon arriving for the consultation, they undergo secure, in-person identity verification (Section \ref{subpar:in_person_identity_verification}).

\paragraph*{Appointment Booking}
\label{appointment_booking_pseudonym_appointment_booking}

To schedule an appointment, the patient selects a time slot at a healthcare organization (e.g., a community health clinic) and sends an appointment scheduling request.

The format of the Appointment Schedule Request is:
$$
\begin{aligned}
\mathsf{AppointmentScheduleRequest} =\ &
\big[ \big(
    \mathsf{PAI},\ 
    \mathsf{AT},\ 
    \mathsf{ScheduleInfo}
\big),\\
&\quad \mathsf{Sig}_{\mathsf{P}_{\mathsf{patient}}} \big]
\end{aligned}
$$

where:
\begin{itemize}
    \item $ \mathsf{PAI} $: Pseudonym Access Information, consisting of 
          $ (\mathsf{P}_{\mathsf{patient}},\ \mathsf{rk}_{\text{patient}\rightarrow\text{HRR}},\ \mathsf{ct}_{\mathsf{PID}}) $
    \item $ \mathsf{AT} $: Appointment Token issued by the APC
    \item $ \mathsf{ScheduleInfo} $: requested appointment details (date, time slot, provider)
    \item $ \mathsf{Sig}_{\mathsf{P}_{\mathsf{patient}}} $: Blind-IBS signature generated using the pseudonym-specific private key $ \mathsf{sk}_{\mathsf{P}_{\mathsf{patient}}} $
\end{itemize}

Next, the healthcare organization presents its GHA-issued $ \mathsf{L\!-\!VC}_{\mathrm{HO}} $ for legitimacy verification through the established mutually authenticated DIDComm channel. Upon successful verification, the patient sends the signed $\mathsf{AppointmentScheduleRequest}$ to the healthcare organization.

Upon receipt, the healthcare organization module:
\begin{itemize}
    \item Verifies $ \mathsf{Sig}_{\mathsf{P}_{\mathsf{patient}}} $ using the AMM of its IDM framework.
    \item Verifies the APC’s signature $ \mathsf{Sig}_{\mathsf{APC}} $ on $ \mathsf{AT} $.
    \item Validates the appointment schedule details.
    \item Queries the ATI Uniqueness Validator (a smart contract on the Appointment Token Usage Ledger) to ensure the $ \mathsf{ATI} $ has not been used. If unused, the validator records it, and the request is accepted; otherwise, the request is rejected as a replay attempt.
    \item Generates a unique appointment confirmation code and returns it to the patient.
    \item Records the booking event in the Auditor Ledger using visit-scoped $ \mathsf{P}_{\mathsf{patient}} $ as $\mathsf{patientIdentifier}$, with relevant $\mathsf{eventDetails}$ such as $ \mathsf{AppointmentConfirmationCode} $ and a single-use, time-bound appointment token reference $ \mathsf{ATI} $.
\end{itemize}

\begin{figure}[h]
    \centering
    \includegraphics[width=8cm, trim={0cm 0cm 0cm 0cm}, clip]{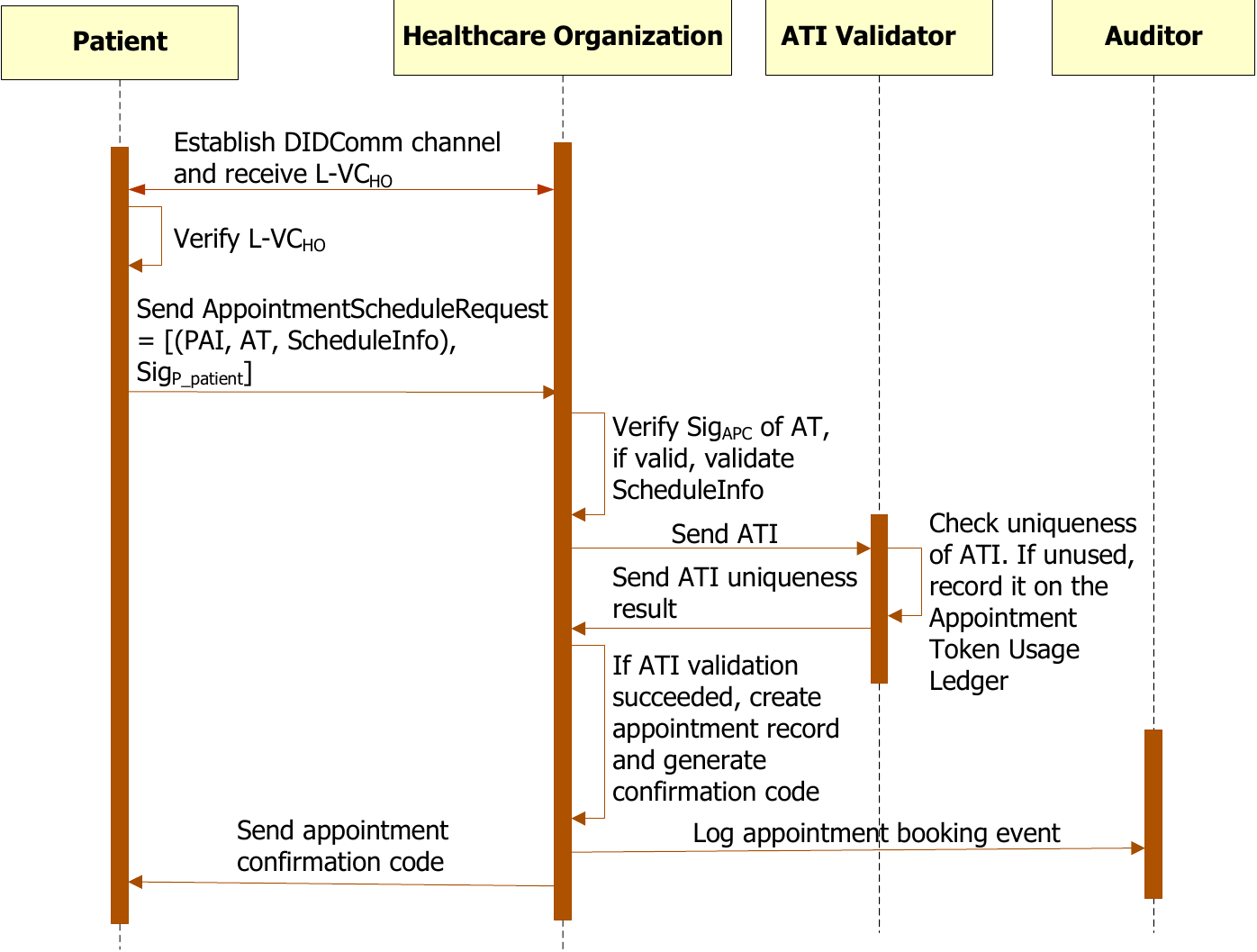}
    \caption{Sequence diagram of appointment booking using pseudonym and Appointment Token}
    \label{fig:appointment_bookings}
\end{figure}

Figure~\ref{fig:appointment_bookings} shows the sequence diagram of appointment booking using PAI and an Appointment Token, highlighting how the healthcare organization verifies the Appointment Token and the ATI.

\paragraph*{In-Person Identity Verification}
\label{subpar:in_person_identity_verification}

At the scheduled appointment time, the patient arrives at the healthcare organization for identity verification.  

The patient may optionally re-verify the Legitimacy Verification Credential $\mathsf{L\!-\!VC}_{\mathrm{HO}}$ to ensure the healthcare organization's DID has not been revoked before establishing a mutually authenticated DIDComm channel.

The patient then submits the following signed identity verification request:

$$
\begin{aligned}
\mathsf{IdentityVerificationRequest} = {} &
\big[
  (\mathsf{AppointmentConfirmationCode},\\
&\quad
   \mathsf{BioHash},\,
   \mathsf{PoK}_{\mathsf{PCred}},\,
   \mathsf{PT}),\,
  \mathsf{Sig}_{\mathsf{P_{patient}}}
\big]
\end{aligned}
$$

where:
\begin{itemize}
    \item $ \mathsf{AppointmentConfirmationCode} $: unique confirmation code generated during appointment booking.
    \item $ \mathsf{BioHash} $: probabilistic hash of the patient's biometric features (e.g., facial template) generated from a live sample at check-in. The $\mathsf{BioHash}$ is session-scoped, used transiently for local verification, and is not reused, stored, or logged to the Auditor Ledger or any backend database.
    \item $\mathsf{PoK}_{\mathsf{PCred}}$: Proof of knowledge of the Patient Credential $ \mathsf{PoK}_{\mathsf{PCred}} $
    \item $ \mathsf{PT} $: Pseudonym Token issued by the PTA
    \item $ \mathsf{Sig}_{\mathsf{P}_{\mathsf{patient}}} $: Blind-IBS signature generated using the pseudonym-specific private key $ \mathsf{sk}_{\mathsf{P}_{\mathsf{patient}}} $
\end{itemize}

Upon receipt, the healthcare organization:
\begin{itemize}
    \item Verifies $ \mathsf{Sig}_{\mathsf{P}_{\mathsf{patient}}} $ on the $\mathsf{IdentityVerificationRequest}$ using the AMM of its IDM framework.
    \item Validates $ \mathsf{AppointmentConfirmationCode} $.
    \item Captures a live image of the patient and compares its probabilistic facial hash with $ \mathsf{BioHash} $ in the request.
    \item Verifies $ \mathsf{PoK}_{\mathsf{PCred}} $ and the digital signature on $ \mathsf{PT} $.
    \item Records the identity verification event in the Auditor Ledger using visit-scoped $ \mathsf{P}_{\mathsf{patient}} $ as $ \mathsf{patientIdentifier} $, together with minimal event metadata such as the $ \mathsf{AppointmentConfirmationCode} $ and a verification outcome (e.g., biometric match successful). No persistent identifiers or biometric data are stored, preventing correlation across visits or healthcare organizations.
\end{itemize}

\begin{figure}[h]
    \centering
    \includegraphics[width=8cm, trim={0cm 0cm 0cm 0cm}, clip]{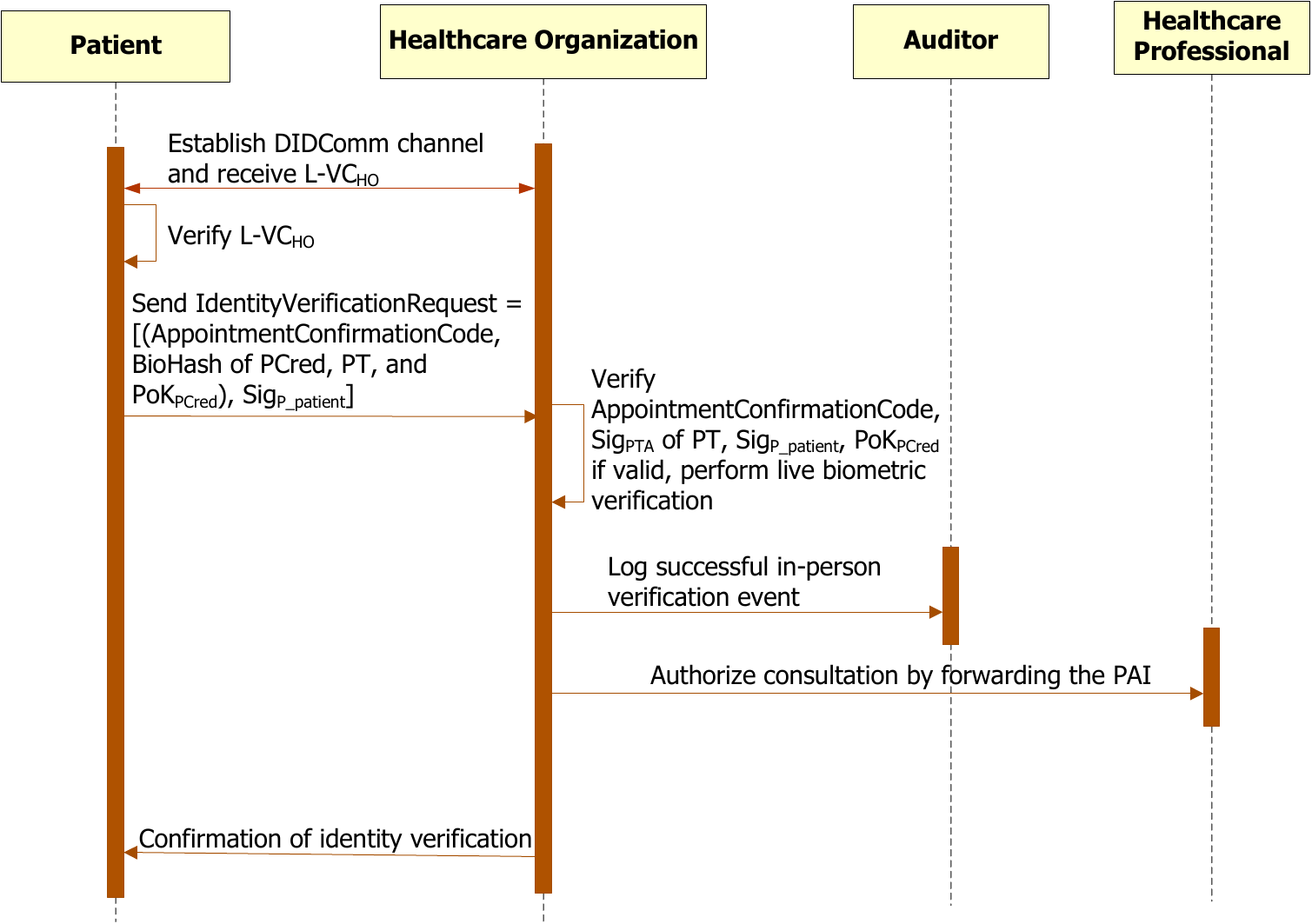}
    \caption{Sequence diagram of in-person identity verification at the healthcare organization}
    \label{fig:in_person_identity_verification}
\end{figure}

If the verification succeeds, the healthcare organization confirms that the individual matches the booked appointment and holds valid credentials: $ \mathsf{PCred} $, $ \mathsf{AT} $, and $ \mathsf{PT} $.

Next, the healthcare organization and the designated healthcare professional exchange DIDs to establish a new secure DIDComm channel for all subsequent in-person interactions. Specifically, the healthcare organization sends $ \mathsf{DID}_{\mathsf{HO}_i} $ to the healthcare professional, who responds with $ \mathsf{DID}_{\mathsf{HP}_i} $. Both parties then present their GHA-issued legitimacy verification credentials $ \mathsf{L\!-\!VC}_{\mathrm{HO}_i} $ for the healthcare organization and $ \mathsf{L\!-\!VC}_{\mathrm{HP}_i} $ for the healthcare professional for mutual trust verification. Finally, the healthcare organization transmits $ \mathsf{P}_{\mathsf{patient}} $, $ \mathsf{rk}_{\text{patient} \rightarrow \mathsf{HRR}} $, and $ \mathsf{ct}_{\text{PID}} $ to the healthcare professional over the authenticated channel and records a corresponding audit entry using only the visit-scoped $ \mathsf{P}_{\mathsf{patient}} $ and minimal event metadata, as defined above.

Figure~\ref{fig:in_person_identity_verification} shows the sequence diagram of in-person identity verification at the healthcare organization, highlighting biometric verification, pseudonym token validation, and proof of Patient Credential possession.

\paragraph{Patient Interaction with Healthcare Professional}
\label{patient_interaction_with_healthcare_professional}

In the consultation phase, the healthcare professional uses the Healthcare Professional Module to request the patient’s longitudinal health record from the HRR. To enable this, the healthcare organization forwards the PAI to the professional, who re-encrypts the patient’s pseudonym $ \mathsf{P}_{\mathsf{patient}} $ with the re-encryption key $ \mathsf{rk}_{\mathsf{patient} \rightarrow \mathsf{HRR}} $, producing the HRR-specific pseudonym $ \mathsf{P}_{\mathsf{HRR}} $ (Section~\ref{subpar:healthcare_proxy_reencryption}).

\begin{figure}[h]
    \centering
    \includegraphics[width=8cm, trim={0cm 0cm 0cm 0cm}, clip]{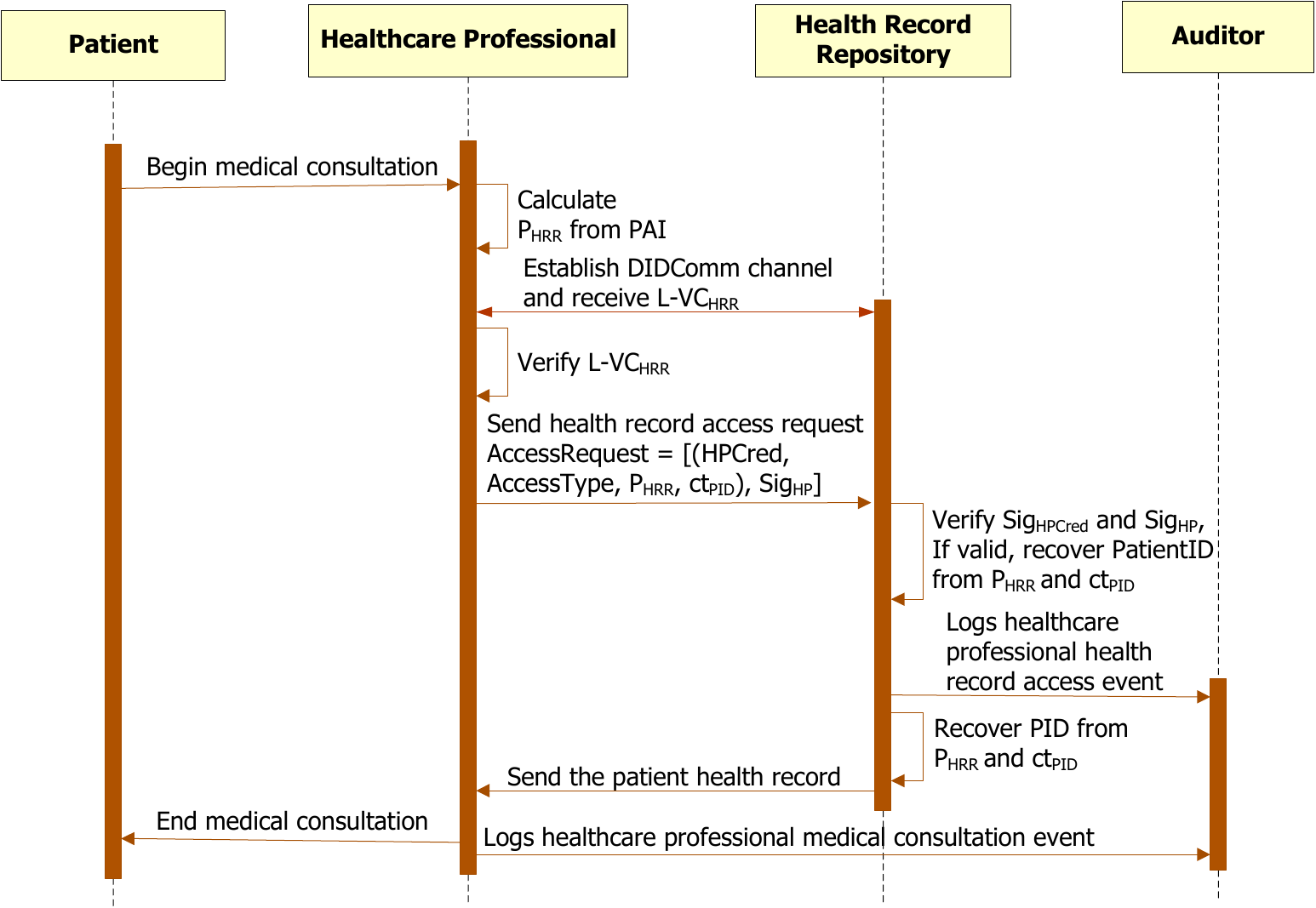}
    \caption{Sequence diagram of patient consultation with the healthcare professional}
    \label{fig:patient_consultation}
\end{figure}

The healthcare professional and the HRR then establish a mutually authenticated DIDComm channel. Through the DIDComm channel, they exchange and verify their GHA-issued legitimacy credentials $ \mathsf{L\!-\!VC}_{\mathrm{HP}_i} $ for the professional and $ \mathsf{L\!-\!VC}_{\mathrm{HRR}} $ for the HRR.  

The healthcare professional then submits the signed access request:
$$
\begin{aligned}
\mathsf{AccessRequest} =\ &
\big[ \big(
    \mathsf{HPCred},\ 
    \mathsf{AccessType},\ 
    \mathsf{P}_{\mathsf{HRR}},\ 
    \mathsf{ct}_{\mathsf{PID}}
\big),\\
\mathsf{Sig}_{\mathsf{HP}} \big]
\end{aligned}
$$

where:
\begin{itemize}
  \item $ \mathsf{HPCred} $: GHA-signed credential authorizing the healthcare professional’s access.
  \item $ \mathsf{AccessType} $: Intended operation on the patient's record (e.g., read or write).
  \item $ \mathsf{P}_{\mathsf{HRR}} $: HRR-specific pseudonym derived by the healthcare professional, enabling patient identification without revealing the patient’s real identity.
  \item $ \mathsf{ct}_{\mathsf{PID}} $: Encrypted healthcare-specific patient identifier.
  \item $ \mathsf{Sig}_{\mathsf{HP}} $: Digital signature using the healthcare professional’s private key $ \mathsf{sk}_{P_{\text{HP}}} $, generated via the Schnorr signature scheme.
\end{itemize}

Upon receiving the request, the HRR:

\begin{itemize}

    \item Verifies $ \mathsf{HPCred} $ using its AMM.
    \item Verifies $ \mathsf{Sig}_{\mathsf{HP}} $.
    \item Decrypts $ \mathsf{P}_{\mathsf{HRR}} $ with $ \mathsf{sk}_{\mathsf{HRR}} $ to recover $ \mathsf{PatientID}_{G_T} $.
    \item Derives the symmetric key $k = \mathsf{HKDF}\big(\mathsf{salt}, \mathsf{enc}(\mathsf{PatientID}_{G_T}), \mathsf{info}\big)$ and decrypts $ \mathsf{ct}_{\mathsf{PID}} $ to obtain the original $ \mathsf{PatientID} $.
    \item Retrieves the corresponding health record using $ \mathsf{PatientID} $.
    \item Records the health record access event in the Auditor Ledger using visit-scoped $ \mathsf{P}_{\mathsf{patient}} $ as $ \mathsf{patientIdentifier} $, together with minimal event metadata such as access type (read/write), healthcare professional identifier, and access timestamp. No persistent identifiers or biometric data are stored, preventing correlation across visits or healthcare organizations (Section~\ref{subsec:architecture-healthcare-idm}).

\end{itemize}

This enables the healthcare professional to review the patient's historical data and, if needed, append new entries such as prescriptions or consultation notes. Figure~\ref{fig:patient_consultation} shows the sequence diagram of a patient consultation with a healthcare professional. It presents how the healthcare professional submits a signed access request to retrieve the patient's health record from the HRR. Table~\ref{tab:scyther-verification} shows that the proposed HIDM framework is resistant to eavesdropping, replay, and impersonation attacks under the Dolev–Yao model.

\section{Security and Privacy Analysis}
\label{sec:security_privacy_analysis}

This section presents a comprehensive evaluation of the proposed \textit{HIDM} framework through two complementary analysis mechanisms that jointly validate its security and privacy properties.

The MSRA-based analysis~\cite{gurses2006contextualizing} evaluates how the framework's architecture and design choices address the diverse and often competing security and privacy goals of healthcare stakeholders. It systematically maps potential threats to the security mechanisms embedded in the HIDM framework, thereby demonstrating its alignment with established security objectives and healthcare regulatory requirements (Section \ref{subsec:msra_security_privacy}).

The formal security and privacy analysis~\cite{cremers2008scyther,meier2013tamarin} provides mathematical proofs for the core cryptographic artifacts employed within the framework. It establishes the unforgeability and replay resistance of issuer-signed artifacts under standard computational assumptions, ensuring the cryptographic soundness of the underlying protocols (Section \ref{subsec:formal_security_privacy_analysis}).

\subsection{Threat Model and Trust Assumptions}

A standard Dolev--Yao network adversary~\cite{dolev2003security} is considered, who can observe, intercept, replay, and modify messages exchanged over public communication channels, but cannot break the underlying cryptographic primitives. In addition, insider adversaries are considered who may control or compromise individual system components while continuing to follow the prescribed protocol execution (honest--but--curious behavior).

The proposed HIDM framework operates across multiple administrative entities with clearly separated roles. Its security relies on the explicit distribution of trust among these entities rather than on any single fully trusted party:

\begin{itemize}

\item \textbf{APC}: Trusted to perform correct identity verification and maintain the mapping between a \emph{PatientID} and real-world PII. The APC is not trusted with access-context information, such as appointment details or healthcare interactions.

\item \textbf{PTA}: Trusted to issue and validate pseudonym tokens correctly and to map PTIs to \emph{PatientIDs} only when authorized. The PTA does not possess PII attributes and cannot independently identify patients.

\item \textbf{HRR}: Treated as honest--but--curious. While it utilizes the PRE scheme under authorized workflows to recover a healthcare-specific $\mathsf{PatientID}$ for longitudinal record indexing, it never has access to the patient’s real-world PII.

\end{itemize}

It is assumed that the APC and PTA do not collude outside of explicitly authorized trace procedures. Collusion between these entities would enable direct resolution from pseudonyms to real-world identities and is therefore considered out of scope for the proposed threat model. Similarly, healthcare organizations are assumed not to collude simultaneously with both the APC and the PTA.
These trust boundaries serve as the foundational constraints for the MSRA-based evaluation in Section~\ref{subsec:msra_security_privacy}, ensuring that security and privacy properties, including conditional traceability, are assessed with respect to the explicit capabilities and limitations of each stakeholder.

\subsection{MSRA-Based Analysis}
\label{subsec:msra_security_privacy}

The analysis proceeds in three stages: 
First, Section~\ref{subsubsec:stakeholder_episode_defs} identifies the stakeholders and defines the functional episodes forming the basis of the evaluation.
Next, Section~\ref{subsubsec:security_privacy_goal_elicitation} 
maps the security and privacy requirements of each stakeholder to the mechanisms implemented in the framework.
Finally, Section~\ref{subsubsec:resolution_mechanism} describes the conflict detection and resolution mechanism that reconciles patient privacy with the need for accountability.

\subsubsection{Stakeholder and Episode Definitions} \label{subsubsec:stakeholder_episode_defs}
This subsection establishes the analytical context for the MSRA-based evaluation by identifying the stakeholders in the HIDM framework and outlining the functional episodes in which they interact.

\paragraph*{Stakeholders}
The HIDM comprises the following stakeholders, each with distinct roles, responsibilities, as well as security and privacy objectives:

\begin{itemize}
    \item \textbf{Patient:} Seeks to control their personal and health information through privacy and informed consent.
    \item \textbf{Healthcare Provider:} Requires authenticated and authorized access to patient health records in order to deliver healthcare services.
    \item \textbf{GHA:} Acts as the root of trust by issuing L-VCs to other components of the IDM framework.
    \item \textbf{APC:} Issues the \textit{Patient Credential} and pseudonym-specific private keys.
    \item \textbf{PTA:} Issues PTs that bind a pseudonym to a patient identifier without revealing the identifier to service providers.
    \item \textbf{Auditor:} Maintains tamper-evident logs to enable traceability and accountability across the ecosystem.
\end{itemize}

\paragraph*{Functional Episodes}
To contextualize the security and privacy requirements for each stakeholder, the following functional episodes are modeled. Each episode represents a discrete interaction within the HIDM framework. The analysis considers the following episodes:

\begin{enumerate}
  \item Issuance of a \(\mathsf{PCred}\)
  \item Issuance of a \(\mathsf{PT}\)
  \item Issuance of a \(sk_{P_{\mathsf{patient}}}\)
  \item Issuance of an \(\mathsf{AT}\)
  \item Appointment booking with a healthcare organization
  \item In-person identity verification at the healthcare organization
  \item Authorization by the healthcare organization for consultation with the healthcare professional
  \item Access by the healthcare professional to the patient's health record from the HRR
\end{enumerate}

\section{Security and Privacy Analysis}
\label{sec:security_privacy_analysis}

This section presents a comprehensive evaluation of the proposed \textit{HIDM} framework through two complementary analysis mechanisms that jointly validate its security and privacy properties.

The MSRA-based analysis~\cite{gurses2006contextualizing} evaluates how the framework's architecture and design choices address the diverse and often competing security and privacy goals of healthcare stakeholders. It systematically maps potential threats to the security mechanisms embedded in the HIDM framework, thereby demonstrating its alignment with established security objectives and healthcare regulatory requirements (Section \ref{subsec:msra_security_privacy}).

The formal security and privacy analysis~\cite{cremers2008scyther,meier2013tamarin} provides mathematical proofs for the core cryptographic artifacts employed within the framework. It establishes the unforgeability and replay resistance of issuer-signed artifacts under standard computational assumptions, ensuring the cryptographic soundness of the underlying protocols (Section \ref{subsec:formal_security_privacy_analysis}). 

\subsection{Threat Model and Trust Assumptions}

A standard Dolev--Yao network adversary~\cite{dolev2003security} is considered, who can observe, intercept, replay, and modify messages exchanged over public communication channels, but cannot break the underlying cryptographic primitives. In addition, insider adversaries are considered who may control or compromise individual system components while continuing to follow the prescribed protocol execution (honest-but-curious behavior).

The proposed HIDM framework operates across multiple administrative entities with clearly separated roles. Its security relies on the explicit distribution of trust among these entities rather than on any single fully trusted party:

\begin{itemize}

\item \textbf{APC}: Trusted to perform correct identity verification and maintain the mapping between a $\mathsf{PatientID}$ and real-world PII. The APC is not trusted with access-context information, such as appointment details or healthcare interactions.

\item \textbf{PTA}: Trusted to issue and validate pseudonym tokens correctly and to map PTIs to $\mathsf{PatientIDs}$ only when authorized. The PTA does not possess PII attributes and cannot independently identify patients.

\item \textbf{HRR}: Treated as honest-but-curious. While it utilizes the PRE scheme under authorized workflows to recover a healthcare-specific $\mathsf{PatientID}$ for longitudinal record indexing, it never has access to the patient’s real-world PII.

\end{itemize}

It is assumed that the APC and PTA do not collude outside of explicitly authorized trace procedures. Collusion between these entities would enable direct resolution from pseudonyms to real-world identities and is therefore considered out of scope for the proposed threat model. Similarly, healthcare organizations are assumed not to collude simultaneously with both the APC and the PTA. These trust boundaries serve as the foundational constraints for the MSRA-based evaluation in Section~\ref{subsec:msra_security_privacy}, ensuring that security and privacy properties, including conditional traceability, are assessed with respect to the explicit capabilities and limitations of each stakeholder.

\subsection{MSRA-Based Analysis}
\label{subsec:msra_security_privacy}

The analysis proceeds in three stages: First, Section~\ref{subsubsec:stakeholder_episode_defs} identifies the stakeholders and defines the functional episodes forming the basis of the evaluation. Next, Section~\ref{subsubsec:security_privacy_goal_elicitation} maps the security and privacy requirements of each stakeholder to the mechanisms implemented in the framework. Finally, Section~\ref{subsubsec:resolution_mechanism} describes the conflict detection and resolution mechanism that reconciles patient privacy with the need for accountability.

\subsubsection{Stakeholder and Episode Definitions} \label{subsubsec:stakeholder_episode_defs}
This subsection establishes the analytical context for the MSRA-based evaluation by identifying the stakeholders in the HIDM framework and outlining the functional episodes in which they interact.

\paragraph*{Stakeholders}
The HIDM comprises the following stakeholders, each with distinct roles, responsibilities, as well as security and privacy objectives:

\begin{itemize}
    \item \textbf{Patient:} Seeks to control their personal and health information through privacy and informed consent.
    \item \textbf{Healthcare Provider:} Requires authenticated and authorized access to patient health records in order to deliver healthcare services.
    \item \textbf{GHA:} Acts as the root of trust by issuing L-VCs to other components of the IDM framework.
    \item \textbf{APC:} Issues the \textit{Patient Credential} and pseudonym-specific private keys.
    \item \textbf{PTA:} Issues PTs that bind a pseudonym to a patient identifier without revealing the identifier to service providers.
    \item \textbf{Auditor:} Maintains tamper-evident logs to enable traceability and accountability across the ecosystem.
\end{itemize}

\paragraph*{Functional Episodes}
To contextualize the security and privacy requirements for each stakeholder, the following functional episodes are modeled. Each episode represents a discrete interaction within the HIDM framework. The analysis considers the following episodes:

\begin{enumerate}
  \item Issuance of a $\mathsf{PCred}$
  \item Issuance of a $\mathsf{PT}$
  \item Issuance of a $sk_{P_{\mathsf{patient}}}$
  \item Issuance of an $\mathsf{AT}$
  \item Appointment booking with a healthcare organization
  \item In-person identity verification at the healthcare organization
  \item Authorization by the healthcare organization for consultation with the healthcare professional
  \item Access by the healthcare professional to the patient's health record from the HRR
\end{enumerate}

\subsubsection{Security and Privacy Goal Elicitation} \label{subsubsec:security_privacy_goal_elicitation}

This subsection maps the stakeholders and episodes to specific security and privacy goals. For each stakeholder in each episode, the goals are expressed using the following evaluation criteria:

\begin{itemize}
    \item \textbf{Information Protected:} The type of data that must be safeguarded within the episode.
    \item \textbf{Counter-Stakeholder:} The entity against whom the security or privacy goal is defined (e.g., adversary or internal stakeholder with conflicting interests).
    \item \textbf{Context:} The specific process or scenario in which the requirement applies.
    \item \textbf{Rationale:} The justification for the requirement.
\end{itemize}

While the framework addresses various operational workflows, the detailed analysis focuses primarily on Episodes 5 and 6. These episodes represent the most privacy-sensitive interactions (i.e., direct patient-provider contact) and serve as the core
scenarios for evaluating the proposed framework. 
The remaining episodes are nevertheless subject to explicit security and privacy requirements. Tables~\ref{tab:msra-summary-a} and~\ref{tab:msra-summary-b} extend the MSRA-based analysis to Episodes E1–E4 and E5–E8, respectively, summarizing stakeholder-specific requirements, protected information, threat actors, and the mechanisms employed to satisfy these requirements across the entire HIDM workflow.

\paragraph{Episode 5: Appointment Booking with the Healthcare Organization}

\medskip
\noindent\textbf{Stakeholder: Patient}

\begin{enumerate}
  \item \textit{Security and Privacy Requirement:} Verify the legitimacy of the healthcare organization ($\mathsf{HO}_i$).\\
  \textbf{Information Protected:} $\mathsf{AppointmentScheduleRequest}$ contents, including $\mathsf{PAI}$, $\mathsf{AT}$, and $\mathsf{ScheduleInfo}$ (without revealing $\mathsf{PatientID}$).\\
  \textbf{Counter-Stakeholder:} Fraudulent $\mathsf{HO}_i$ (impersonation).\\
  \textbf{Context:} Pre-booking event.\\
  \textbf{Rationale:} Prevent the disclosure of sensitive artifacts to malicious entities.\\
  \textbf{Fulfilling Mechanism:} Verify $\mathsf{L\!-\!VC}_{\mathsf{HO}_i}$ and use a mutually authenticated DIDComm channel.
\end{enumerate}

\medskip
\noindent\textbf{Stakeholder: Healthcare Organization ($\mathsf{HO}_i$)}

\begin{enumerate}
  \item \textit{Security and Privacy Requirement:} Authenticate and validate the appointment request.\\
  \textbf{Information Protected:} $\mathsf{AppointmentScheduleRequest}$ (integrity and authenticity).\\
  \textbf{Counter-Stakeholder:} Patient without a valid pseudonym-specific private key.\\
  \textbf{Context:} Booking event.\\
  \textbf{Rationale:} Ensure the requester controls $P_{\mathsf{patient}}$ via $sk_{P_{\mathsf{patient}}}$.\\
  \textbf{Fulfilling Mechanism:} Verify $\mathsf{Sig}_{\mathsf{P_{patient}}}$ on $\mathsf{AppointmentScheduleRequest}$.

  \item \textit{Security and Privacy Requirement:} Verify the authenticity and validity of $\mathsf{AT}$.\\
  \textbf{Information Protected:} $\mathsf{AT}$ (integrity and issuer authenticity).\\
  \textbf{Counter-Stakeholder:} Forged or expired $\mathsf{AT}$.\\
  \textbf{Context:} Booking event.\\
  \textbf{Rationale:} Prevent fraudulent or stale bookings.\\
  \textbf{Fulfilling Mechanism:} Verify $\mathsf{Sig}_{\mathsf{APC}}$ on $\mathsf{AT}$ and check $\mathsf{exp}$.

  \item \textit{Security and Privacy Requirement:} Enforce the single-use property of $\mathsf{AT}$.\\
  \textbf{Information Protected:} $\mathsf{ATI}$ uniqueness.\\
  \textbf{Counter-Stakeholder:} A patient attempting to replay or reuse $\mathsf{AT}$.\\
  \textbf{Context:} Booking event.\\
  \textbf{Rationale:} Prevent multiple bookings from a single token.\\
  \textbf{Fulfilling Mechanism:} Check the $\mathsf{ATI}$ Usage Ledger prior to booking.
\end{enumerate}

\bigskip

\paragraph{Episode 6: In-Person Identity Verification at the Healthcare Organization}

\medskip
\noindent\textbf{Stakeholder: Patient}

\begin{enumerate}
  \item \textit{Security and Privacy Requirement:} Minimize disclosure from $\mathsf{PCred}$.\\
  \textbf{Information Protected:} Only $\mathsf{BioHash}$ and required proofs.\\
  \textbf{Counter-Stakeholder:} Over-collection by $\mathsf{HO}$.\\
  \textbf{Context:} Check-in event.\\
  \textbf{Rationale:} Prevent unnecessary attribute leakage and profiling.\\
  \textbf{Fulfilling Mechanism:} Selective disclosure via CL signatures with $\mathsf{PoK}_{\mathsf{PCred}}$.
\end{enumerate}

\medskip
\noindent\textbf{Stakeholder: Healthcare Organization}

\begin{enumerate}
  \item \textit{Security and Privacy Requirement:} Verify selective disclosure from $\mathsf{PCred}$.\\
  \textbf{Information Protected:} Validity of disclosed attributes (e.g., $\mathsf{BioHash}$).\\
  \textbf{Counter-Stakeholder:} Fraudulent patient presenting invalid $\mathsf{PCred}$.\\
  \textbf{Context:} Check-in event.\\
  \textbf{Rationale:} Confirm legitimate credential holder and attribute validity.\\
  \textbf{Fulfilling Mechanism:} Verify NIZK proof of $\mathsf{PCred}$ possession.

  \item \textit{Security and Privacy Requirement:} Validate the $\mathsf{PT}$.\\
  \textbf{Information Protected:} $\mathsf{PT}$ (authenticity and validity).\\
  \textbf{Counter-Stakeholder:} Forged or tampered $\mathsf{PT}$.\\
  \textbf{Context:} Check-in event.\\
  \textbf{Rationale:} Admit only patients with legitimate pseudonym tokens.\\
  \textbf{Fulfilling Mechanism:} Verify $\mathsf{Sig}_{\mathsf{PTA}}$ on $\mathsf{PT}$.

  \item \textit{Security and Privacy Requirement:} Ensure the authenticity of the identity-verification request.\\
  \textbf{Information Protected:} Integrity of the request.\\
  \textbf{Counter-Stakeholder:} Patient without a valid pseudonym-specific key.\\
  \textbf{Context:} Check-in event.\\
  \textbf{Rationale:} Guarantee control of $P_{\mathsf{patient}}$ via $sk_{P_{\mathsf{patient}}}$.\\
  \textbf{Fulfilling Mechanism:} Verify $\mathsf{Sig}_{P_{\mathsf{patient}}}$.
\end{enumerate}

\begin{table*}[!t]
\scriptsize
\setlength{\tabcolsep}{3pt}
\renewcommand{\arraystretch}{1.12}
\centering
\begin{tabular}{|p{1.1cm}|p{1.7cm}|p{3.5cm}|p{2.4cm}|p{2.6cm}|p{3.3cm}|}
\hline
\textbf{Episode} & \textbf{Stakeholder} & \textbf{Requirements} & \textbf{Information Protected} & \textbf{Threat Actor} & \textbf{Mechanism} \\
\hline

\multirow{2}{*}{\textbf{E1}} & Patient &
Secure channel for PCred issuance; verify APC legitimacy; ensure PCred is verifiable &
PII during issuance; PCred integrity &
Eavesdropper; Fraudulent APC; Forger/Tamperer &
DIDComm (mutual auth); Verify L-VC\textsubscript{APC}; CL-signature verification \\
\cline{2-6}
& APC &
Protect mapping PII$\leftrightarrow$PatientID; record issuance event (minimal metadata) &
Mapping data; Issuance metadata &
Unauthorized accessor; Tampering adversary &
Hardened DB; Auditor ledger logging \\
\hline

\multirow{2}{*}{\textbf{E2}} & Patient &
Secure channel; verify PTA legitimacy; PT independently verifiable &
PatientID; $\mathsf{PoK}_{\mathsf{PCred}}$; PBP; PT &
Eavesdropper; Fraudulent PTA; Forger/Tamperer &
DIDComm; Verify L-VC\textsubscript{PTA}; Schnorr-signature verification \\
\cline{2-6}
& PTA &
Issue PT only to valid PCred holders; verify PBP (pseudonym$\leftrightarrow$PatientID); protect mapping; record issuance &
Credential integrity; Mapping data; Issuance metadata &
Invalid requester; Unauthorized accessor; Tampering adversary &
CL-sig verify + $\mathsf{PoK}_{\mathsf{PCred}}$; PBP (NIZK) verify; Secure DB; Auditor ledger \\
\hline

\multirow{2}{*}{\textbf{E3}} & Patient &
Secure channel; verify APC legitimacy; maintain unlinkability of $sk_{P_{\mathsf{patient}}}$ from $P_{\mathsf{patient}}$ and $\mathsf{PatientID}$ &
$\mathsf{PatientID}$; Blinded pseudonym $P'_{\mathsf{patient}}$; $sk_{P_{\mathsf{patient}}}$ &
Eavesdropper; Fraudulent APC; Curious issuer &
DIDComm; Verify L-VC\textsubscript{APC}; Blind-IBS key extraction \\
\cline{2-6}
& APC &
Issue $sk_{P_{\mathsf{patient}}}$ only to valid $\mathsf{PCred}$ holders; record key issuance (no true pseudonym) &
Credential integrity; Issuance metadata (incl. blinded pseudonym) &
Invalid requester; Tampering adversary &
CL-sig verify + $\mathsf{PoK}_{\mathsf{PCred}}$; Auditor ledger \\
\hline

\multirow{2}{*}{\textbf{E4}} & Patient &
Secure channel; verify APC legitimacy; AT independently verifiable; ATI is patient-generated and hidden from APC &
ATI (hidden); exp; AT integrity &
Eavesdropper; Fraudulent APC; Forger/Tamperer &
DIDComm; Verify L-VC\textsubscript{APC}; Partially blind Schnorr (ATI blinded); CL-sig verify + $\mathsf{PoK}_{\mathsf{PCred}}$ \\
\cline{2-6}
& APC &
Issue AT only to valid PCred; validate blinded ATI commitment; enforce ATI uniqueness; record issuance &
Blinded-ATI binding; Single-use property; Issuance metadata &
Invalid requester; Inconsistent commit; Replay attacker; Tampering adversary &
CL-sig verify + $\mathsf{PoK}_{\mathsf{PCred}}$; Partially blind Schnorr protocol; ATI Usage Ledger; Auditor ledger \\
\hline

\end{tabular}
\caption{MSRA-based analysis of HIDM (Part I: Episodes E1–E4)}
\label{tab:msra-summary-a}
\end{table*}

\begin{table*}[!t]
\scriptsize
\setlength{\tabcolsep}{3pt}
\renewcommand{\arraystretch}{1.12}
\centering
\begin{tabular}{|p{1.1cm}|p{1.7cm}|p{3.5cm}|p{2.4cm}|p{2.6cm}|p{3.3cm}|}
\hline
\textbf{Episode} & \textbf{Stakeholder} & \textbf{Requirements} & \textbf{Information Protected} & \textbf{Threat Actor} & \textbf{Mechanism} \\
\hline

\multirow{2}{*}{\textbf{E5}} & Patient &
Securely send \textit{AppointmentScheduleRequest} (PAI, AT, ScheduleInfo) and receive confirmation; verify HO legitimacy &
Request content (PAI, AT, ScheduleInfo); Confirmation code &
Eavesdropper; Fraudulent HO &
DIDComm; Verify L-VC\textsubscript{HO$_i$} \\
\cline{2-6}
& HO$_i$ &
Verify Sig\textsubscript{$P_{\text{patient}}$}; verify AT (Sig\textsubscript{APC}, check exp); enforce single-use ATI; protect stored PAI; record booking &
Request/Token integrity; ATI; PAI (at rest); Booking metadata &
Fraudulent requester; Token forgery/misuse; Replay; Tampering adversary &
Schnorr-verify (Sig\textsubscript{$P_{\text{patient}}$}); Verify partially blind Schnorr on AT (Sig\textsubscript{APC}); ATI Usage Ledger; Secure DB; Auditor ledger \\
\hline

\multirow{2}{*}{\textbf{E6}} & Patient &
Securely send \textit{IdentityVerificationRequest}; minimal disclosure from PCred &
BioHash; PT; $\mathsf{PoK}_{\mathsf{PCred}}$; Confirmation code; Undisclosed attrs &
Eavesdropper; Over-collection by HO &
DIDComm; CL selective disclosure + PoK \\
\cline{2-6}
& HO &
Verify $\mathsf{PoK}_{\mathsf{PCred}}$; validate PT (Sig\textsubscript{PTA}); verify Sig\textsubscript{$P_{\text{patient}}$}; match to booking (live face vs BioHash; confirmation); record event (minimal metadata) &
Binding to valid PCred; Token authenticity; Request integrity; Match evidence; Verification metadata (code, timestamp) &
Impersonator; Token forger; Fraudulent requester; Tampering adversary &
NIZK verify; Schnorr-verify; Liveness + BioHash check; Auditor ledger \\
\hline

\multirow{2}{*}{\textbf{E7}} & HO &
Verify HP legitimacy; transmit PAI securely with least retention &
PAI ($P_{\mathsf{patient}}$, $rk_{\mathsf{patient}\rightarrow \mathsf{HRR}}$, $ct_{\mathsf{PID}}$) &
Eavesdropper; Fraudulent HP; Malicious insider &
DIDComm; Verify L-VC\textsubscript{HP$_i$}; Least-retention policy \\
\cline{2-6}
& HP &
Verify HO legitimacy; log authorization without linkability &
PAI received from HO; Authorization metadata (event ID, timestamp) &
Fraudulent HO; Tampering adversary &
Verify L-VC\textsubscript{HO$_i$}; Auditor ledger \\
\hline

\multirow{2}{*}{\textbf{E8}} & HP &
Securely send \textit{AccessRequest} (HPCred, AccessType, $P_{\mathsf{HRR}}$, $ct_{\mathsf{PID}}$); verify HRR legitimacy &
AccessRequest contents; Retrieved records &
Eavesdropper; Fraudulent HRR &
DIDComm; Verify L-VC\textsubscript{HRR} \\
\cline{2-6}
& HRR &
Authenticate/authorize HP and enforce scope; validate patient reference and decrypt identifiers; record access without linkability &
Patient health records; Correct resolution of $P_{\mathsf{HRR}}$, $ct_{\mathsf{PID}}$; Access metadata (request ID, timestamp) &
Fraudulent or over-privileged HP; Tampering adversary &
Verify L-VC\textsubscript{HP$_i$} and $\mathsf{Sig}_{\mathsf{HP}}$; Check HPCred scope vs AccessType; Decrypt $P_{\mathsf{HRR}}$ and $ct_{\mathsf{PID}}$; Auditor ledger \\
\hline

\end{tabular}
\caption{MSRA-based analysis of HIDM (Part II: Episodes E5–E8)}
\label{tab:msra-summary-b}
\end{table*}

\subsubsection{Resolution Mechanism: Conditional Traceability}\label{subsubsec:resolution_mechanism}

The proposed HIDM framework incorporates a conditional traceability mechanism that safeguards patient privacy during routine operations while enabling identity reconstruction under exceptional, legally sanctioned conditions. This balance is achieved through a separation-of-duties approach, in which control over identity linkage data is intentionally divided between independent authorities. The following discussion explains how this approach structures the management of linkage data and governs the process for invoking traceability, and then analyzes the resulting security and privacy properties.

\paragraph*{Separation-of-Duties Approach and Traceability Workflow}

Within the healthcare ecosystem, identity-related mappings are partitioned as follows:
\begin{itemize}
    \item \textbf{APC}: Maintains $\{\mathsf{PII} \Leftrightarrow \mathsf{PatientID}\}$, but has no knowledge of pseudonyms $\mathsf{P_{patient}}$.
    \item \textbf{PTA}: Maintains $\{\mathsf{PatientID} \Leftrightarrow \mathsf{P_{patient}}\}$, but has no access to patient $\mathsf{PII}$.
    \item \textbf{Healthcare Organizations and Professionals}: Operate exclusively with $\mathsf{P_{patient}}$ for appointments, treatment, and record access; they never handle $\mathsf{PatientID}$ or $\mathsf{PII}$.
\end{itemize}

Traceability is triggered only upon receipt of a valid legal mandate (e.g., a court order or regulatory directive). A designated authority compels both the APC and PTA to disclose their respective mappings. All disclosures are logged by the Auditor to ensure transparency and prevent misuse. The identity reconstruction process is as follows:

\begin{enumerate}
    \item PTA reveals $\mathsf{P_{patient}} \rightarrow \mathsf{PatientID}$.
    \item APC reveals $\mathsf{PatientID} \rightarrow \mathsf{PII}$.
    \item Combining both reconstructs $\{\mathsf{PII} \Leftrightarrow \mathsf{PatientID} \Leftrightarrow \mathsf{P_{patient}}\}$, linking the pseudonymized activity to the individual.
\end{enumerate}

\subsection{Formal Security and Privacy Analysis}
\label{subsec:formal_security_privacy_analysis}

Having established the architectural soundness of the HIDM framework through the MSRA-based evaluation, this subsection presents its formal validation under well-defined adversarial models. The objective is to verify, with mathematical rigor, that the cryptographic constructions and protocol workflows preserve both security (Section~\ref{subsubsec:formal_security_analysis}) and privacy (Section~\ref{subsubsec:formal_privacy_analysis}) properties even in the presence of an active network adversary.

\subsubsection{Formal Security Analysis}
\label{subsubsec:formal_security_analysis}
The formal security analysis of the HIDM framework is performed in two stages: computational proof-based reasoning and automated symbolic verification. The computational proof-based reasoning establishes the intrinsic cryptographic soundness of the HIDM framework's core artifacts (e.g., credentials, tokens, signed payloads) by proving properties such as unforgeability under standard hardness assumptions and replay resistance under bounded adversarial models (Section ~\ref{paragraph:comp-proof-reasoning}). The automated symbolic verification validates that the protocol workflows constructed from these artifacts maintain end-to-end security properties (e.g., message confidentiality, synchronization, and integrity) in an adversarial communication environment (Section ~\ref{paragraphc:auto-symbolic-verification}).

\paragraph{Computational proof-based reasoning} \label{paragraph:comp-proof-reasoning}

Theorem~\ref{thm:generic-unforgeability} establishes the cryptographic soundness of issuer-signed artifacts under standard computational assumptions. In contrast, Theorem~\ref{thm:replay-resistance-AT} analyzes protocol-level resistance to active network adversaries in the Dolev-Yao model.

\begin{theorem}[Generic unforgeability of issuer-signed artifacts]
\label{thm:generic-unforgeability}

Let $\Sigma$ be a digital signature scheme that is Existentially Unforgeable under Chosen Message Attack (EUF-CMA) secure. Consider any protocol that issues artifacts of the form:
$$
\mathsf{A} = [M,\sigma]\qquad\text{with}\qquad \sigma \leftarrow \mathrm{Sig}_{\mathsf{issuer}}(M)
$$
Then artifacts produced under $\Sigma$ are existentially unforgeable: no adversary can output $[M^{*},\sigma^{*}]$ with $M^{*}$ not previously submitted to the signing oracle and $\mathsf{Verify}(Y_{\mathsf{issuer}},M^{*},\sigma^{*})=1$, except with negligible probability.
\end{theorem}

\begin{proof}[Proof (reduction)]
Assume, for contradiction, a PPT adversary $\mathcal{A}$ forges a valid artifact with advantage $\varepsilon$. We construct a PPT algorithm $\mathcal{B}$ that breaks EUF-CMA security of $\Sigma$ with the same advantage.

\emph{Setup.}
The EUF-CMA challenger $\mathcal{C}_{\Sigma}$ defines the security experiment and provides signing-oracle access. It generates a key pair $(x_{\mathsf{issuer}}, Y_{\mathsf{issuer}})$, gives the public key $Y_{\mathsf{issuer}}$ to $\mathcal{B}$, and keeps the secret key $x_{\mathsf{issuer}}$ private. Algorithm $\mathcal{B}$ then forwards $Y_{\mathsf{issuer}}$ to the adversary $\mathcal{A}$.

\emph{Query simulation.}
Whenever $\mathcal{A}$ requests a signature on a message $M$ (either through a standard or blinded signing interface), $\mathcal{B}$ forwards $M$ to its signing oracle at $\mathcal{C}_{\Sigma}$. The oracle responds with $\sigma = \mathrm{Sig}_{\mathsf{issuer}}(M)$, which $\mathcal{B}$ then relays to $\mathcal{A}$.

\emph{Forgery output.}
Eventually, $\mathcal{A}$ produces a forged signature $\sigma^{*} = \mathrm{Sig}_{\mathsf{issuer}}(M^{*})$ on a message $M^{*}$ that was never submitted to the signing oracle.
 
\emph{Conclusion.}
Since $\Sigma$ is assumed to be EUF-CMA secure, the existence of an adversary $\mathcal{A}$ that forges a valid artifact with non-negligible probability would contradict this assumption. If such an adversary $\mathcal{A}$ existed with success probability $\varepsilon$, then algorithm $\mathcal{B}$ could break the EUF-CMA security of $\Sigma$ with the same advantage. Therefore, all issuer-signed core artifacts of the HIDM framework are unforgeable under standard computational assumptions.
\end{proof}

\begin{corollary}[Unforgeability of domain-specific artifacts] \label{cor:instantiations}
Instantiating Theorem~\ref{thm:generic-unforgeability} yields:
\begin{enumerate}
  \item \emph{Legitimacy credentials for healthcare organizations.} If $\Sigma$ is RSA as used for $\mathsf{L\!-\!VC}_{\mathrm{HO}}$, then $\mathsf{L\!-\!VC}_{\mathrm{HO}}$ is EUF-CMA unforgeable.
  \item \emph{Patient credentials.} If $\Sigma$ is the CL signature used for $\mathsf{PCred}$, then $\mathsf{PCred}$ is EUF-CMA unforgeable.
  \item \emph{Appointment tokens.} If $\Sigma$ is the partially blind Schnorr signature used for $\mathsf{AT}$, then $\mathsf{AT}$ is EUF-CMA unforgeable.
  \item \emph{Pseudonym tokens.} If $\Sigma$ is the Schnorr signature used for $\mathsf{PT}$, then $\mathsf{PT}$ is EUF-CMA unforgeable.
  \item \emph{Patient-signed payloads.} If $\Sigma$ is the Blind-IBS scheme used for patient payloads, then such payloads are EUF-CMA unforgeable.
\end{enumerate}
\end{corollary}

\begin{theorem}[Replay resistance of appointment tokens]
\label{thm:replay-resistance-AT}

In the Dolev-Yao model, the appointment-booking protocol is replay-resistant against any PPT adversary $\mathcal{A}$, provided the healthcare organization (HO) enforces:
(i) one-time redemption via an \emph{ATI uniqueness} check, and
(ii) token \emph{expiration} verification.
\end{theorem}

\begin{proof}[Proof sketch]
By Theorem~\ref{thm:generic-unforgeability}, $\mathcal{A}$ cannot forge a fresh, valid token. It remains to consider a replay of an issued token.
$$
\mathsf{AT}=\big[(\mathsf{ATI},\mathit{exp}),\,\mathrm{Sig}_{\mathsf{APC}}\big]
$$
(1) \emph{Uniqueness.} Before acceptance, the HO queries the \emph{ATI Uniqueness Validator}. If $\mathsf{ATI}$ is already marked redeemed, the request is rejected, eliminating verbatim transcript replays.

(2) \emph{Expiration.} The HO enforces $\textit{now}\le \mathit{exp}$ (with bounded clock skew). Any attempt outside the validity window is rejected.

Thus, a replay can succeed only if the HO fails to apply either check, or if $\mathcal{A}$ modifies $(\mathsf{ATI},\mathit{exp})$ without invalidating $\mathrm{Sig}_{\mathsf{APC}}$, which would contradict Theorem~\ref{thm:generic-unforgeability}. Therefore, the replay advantage is negligible:
$$
\Pr\big[\mathcal{A}\ \text{wins Replay Game}\big]\le \negl(\lambda)
\qedhere
$$
\end{proof}

\paragraph{Automated symbolic verification} 
\label{paragraphc:auto-symbolic-verification}

The verification was conducted using the Scyther tool~\cite{cremers2008scyther}, which models the adversary under the Dolev-Yao abstraction, assuming the perfect cryptographic soundness of the HIDM framework’s core artifacts while granting the adversary complete control over the communication network, including the ability to intercept, modify, inject, and replay messages. Table~\ref{tab:scyther-verification} shows that the proposed HIDM framework is resistant to eavesdropping, replay, and impersonation attacks under the Dolev-Yao model.

\begin{table}[h!]
\centering
\caption{Summary of Scyther Verification Findings for Episodes 5, 6}
\label{tab:scyther-verification}
\begin{tabular}{|p{1.8cm}|p{4.0cm}|p{1.6cm}|}
\hline
\textbf{Security Property Verified} & \textbf{Interpretation of Result} & \textbf{Attack Prevented} \\
\hline
Confidentiality (via claim \texttt{Secret}) &
The \texttt{Secret} claim ensures that sensitive data exchanged during the protocol, such as the appointment payload, pseudonym access information (PAI), and tokens, remains confidential and cannot be derived by an unauthorized adversary. &
Eavesdropping, information leakage \\ 
\hline
Synchronization (via claim \texttt{Nisynch}) &
The non-injective synchronization claim guarantees that both communicating parties (e.g., a patient and a healthcare organization) agree on the identity of their peer and the protocol data. This prevents desynchronization between sessions. &
Man-in-the-Middle (MITM), replay \\ 
\hline
Authenticity \& Integrity (via claims \texttt{Commit-Running}) &
The \texttt{Running} claim asserts that a protocol instance has started with a specific partner, and the corresponding \texttt{Commit} claim verifies that the protocol has completed with the intended partner and consistent data. Joint satisfaction of these claims confirms mutual authentication and message integrity. &
Impersonation, message tampering \\ 
\hline
\end{tabular}
\end{table}

\subsubsection{Formal privacy analysis}
\label{subsubsec:formal_privacy_analysis}

The privacy properties of the HIDM framework were verified using the Tamarin Prover~\cite{meier2013tamarin}, focusing on privacy-specific goals that extend beyond message-level security. Table~\ref{tab:tamarin-privacy} summarizes the verification results, showing that the HIDM framework upholds pseudonymous authentication, prevents linkability across patient sessions, and enables controlled identity reconstruction only under lawful audit conditions. Specifically, the conditional traceability property was formally verified in Tamarin by modeling the separation of duties between the APC and PTA. The analysis proved that an adversary, even with access to the public outputs of both authorities, cannot derive the link between a PII and a pseudonym unless a designated Warrant action, modeling a legal mandate, is triggered to authorize the disclosure of the internal mappings.

\begin{table}[h!]
\centering
\caption{Tamarin Privacy Verification Findings for Episodes 5, 6}
\label{tab:tamarin-privacy}
\begin{tabular}{|p{2.5cm}|p{4.8cm}|}
\hline
\textbf{Privacy Property Verified} & \textbf{Interpretation of Result} \\ \hline
\textbf{Pseudonymous Authentication} & The patient interacts with the healthcare organization using a pseudonym, ensuring that real identifiers remain hidden from service providers and observers. \\ \hline
\textbf{Intra-person Unlinkability} & Multiple actions (bookings or verifications) by the same patient employ fresh pseudonyms, preventing correlation of sessions by an adversary. \\ \hline
\textbf{Inter-person Unlinkability} & A transaction initiated by one patient is observationally indistinguishable from that of another, preventing cross-user linkage. \\ \hline
\textbf{Conditional Traceability} & A lawful audit trace exists where authorized entities can cooperatively reconstruct the patient's identity mapping, while unilateral or unauthorized attempts remain infeasible. \\ \hline
\end{tabular}
\end{table}

\section{Performance Measurement}
\label{sec:performance-measurement}

The performance impact of privacy-preserving credentials within the proposed HIDM framework is evaluated through processing time measurements. Six representative scenarios, described in Section~\ref{subsec:evaluation-methodology}, capture core interactions in the framework. The simulation environment, execution procedure, and resulting metrics are presented in Section~\ref{subsec:simulation-setup-results}.

\subsection{Evaluation Methodology}
\label{subsec:evaluation-methodology}

Three credential signature schemes are considered: baseline RSA, CL-RSA, and CL-Bilinear. 
The RSA scheme is evaluated only for the Patient Credential Issuance scenario, serving as a non-privacy-preserving baseline. 
CL-RSA is included as a representative privacy-preserving scheme based on RSA, allowing a direct comparison with the bilinear variant. 
In contrast, CL-Bilinear credentials leverage pairing-based operations and are evaluated across all scenarios to demonstrate the efficiency gains of the proposed framework. 
Both CL-RSA and CL-Bilinear credentials are carried forward into subsequent interactions (Pseudonym Token issuance, pseudonym-specific key issuance, Appointment Token issuance, Appointment Booking, and In-person Verification), so the measurements reflect the end-to-end performance impact of using CL-based credentials within the IDM framework. 
All schemes are implemented at an equivalent cryptographic strength corresponding to 128-bit security, in accordance with the National Institute of Standards and Technology (NIST) guidelines~\cite{NIST:online} (Table~\ref{tab:nist_key_sizes}), ensuring a fair and consistent basis for comparison.
To ensure statistical reliability, each scenario was executed twelve times. The highest and lowest execution times were discarded, and the average processing time was computed from the remaining ten runs to establish a stable performance baseline, mitigating the effects of cold starts and system caching.
The evaluation covers six scenarios representing the core credential and token operations in the HIDM workflow:

\begin{enumerate}

    \item \textbf{Patient Credential Issuance}: Generate and sign a Patient Credential using signature scheme (RSA, CL-RSA, CL-Bilinear).
    
    \item \textbf{Pseudonym Token Issuance}: Using the chosen credential signature scheme, generate and verify a \( \mathsf{PoK}_{\mathsf{PCred}} \) for a disclosed attribute (e.g., \texttt{PatientID}), then sign the token with a Schnorr signature scheme.
    
    \item \textbf{Pseudonym-Specific Private Key Issuance}: Using the chosen credential signature scheme, generate and verify \( \mathsf{PoK}_{\mathsf{PCred}} \) for a disclosed attribute (e.g., \texttt{PatientID}), then execute the Blind IBS key-extraction protocol to issue the patient's pseudonym-specific private key.

    \item \textbf{Appointment Token Issuance}: Using the chosen credential signature scheme (CL-RSA or CL-Bilinear), generate and verify a \( \mathsf{PoK}_{\mathsf{PCred}} \) for a disclosed attribute (e.g., \texttt{PatientID}), then sign the token with a partially blind Schnorr scheme.
    
   \item \textbf{Appointment Booking}: Using the chosen credential signature scheme, generate and verify the patient's Blind IBS signature on an \texttt{Appointment\allowbreak Schedule\allowbreak Request}, and verify the AT signature.

    \item \textbf{In-person Verification}: Using the chosen credential signature scheme, generate and verify the patient's Blind IBS signature on an \texttt{Identity\allowbreak Verification\allowbreak Request}, generate and verify 
    $\mathsf{PoK}_{\mathsf{PCred}}$ for a disclosed attribute (e.g., \texttt{BioHash}), 
    and verify the PT signature.

\end{enumerate}

\begin{table}[h]
    \centering
    \caption{NIST Recommended Equivalent Key Sizes at 128-bit Security Strength}
    \label{tab:nist_key_sizes}
    \begin{tabular}{ll}
        \hline
        \textbf{Algorithm} & \textbf{Key Size (bits)} \\
        \hline
        RSA Key Size & 3072 \\
        ECC Key Size & 256 \\
        Schnorr Key Size: $p$ (modulus) & 3072 \\
        Schnorr Key Size: $q$ (subgroup order) & 256 \\
        \hline
    \end{tabular}
\end{table}

\section{Performance Measurement}
\label{sec:performance-measurement}

The performance impact of privacy-preserving credentials within the proposed HIDM framework is evaluated through processing time measurements. Six representative scenarios, described in Section~\ref{subsec:evaluation-methodology}, capture core interactions in the framework. The simulation environment, execution procedure, and resulting metrics are presented in Section~\ref{subsec:simulation-setup-results}.

\subsection{Evaluation Methodology}
\label{subsec:evaluation-methodology}

Three credential signature schemes are considered: baseline RSA, CL-RSA, and CL-Bilinear. 
The RSA scheme is evaluated only for the Patient Credential Issuance scenario, serving as a non-privacy-preserving baseline. 
CL-RSA is included as a representative privacy-preserving scheme based on RSA, allowing a direct comparison with the bilinear variant. 
In contrast, CL-Bilinear credentials leverage pairing-based operations and are evaluated across all scenarios to demonstrate the efficiency gains of the proposed framework. 
Both CL-RSA and CL-Bilinear credentials are carried forward into subsequent interactions (Pseudonym Token issuance, pseudonym-specific key issuance, Appointment Token issuance, Appointment Booking, and In-person Verification), so the measurements reflect the end-to-end performance impact of using CL-based credentials within the IDM framework. 

All schemes are implemented at an equivalent cryptographic strength corresponding to 128-bit security, in accordance with the National Institute of Standards and Technology (NIST) guidelines~\cite{NIST:online} (Table~\ref{tab:nist_key_sizes}), ensuring a fair and consistent basis for comparison. To ensure statistical reliability, each scenario was executed twelve times. The highest and lowest execution times were discarded, and the average processing time was computed from the remaining ten runs to establish a stable performance baseline, mitigating the effects of cold starts and system caching. 

The evaluation covers six scenarios representing the core credential and token operations in the HIDM workflow:

\begin{enumerate}

    \item \textbf{Patient Credential Issuance}: Generate and sign a Patient Credential using signature scheme (RSA, CL-RSA, CL-Bilinear).
    
    \item \textbf{Pseudonym Token Issuance}: Using the chosen credential signature scheme, generate and verify a $ \mathsf{PoK}_{\mathsf{PCred}} $ for a disclosed attribute (e.g., \texttt{PatientID}), then sign the token with a Schnorr signature scheme.
    
    \item \textbf{Pseudonym-Specific Private Key Issuance}: Using the chosen credential signature scheme, generate and verify $ \mathsf{PoK}_{\mathsf{PCred}} $ for a disclosed attribute (e.g., \texttt{PatientID}), then execute the Blind IBS key-extraction protocol to issue the patient's pseudonym-specific private key.

    \item \textbf{Appointment Token Issuance}: Using the chosen credential signature scheme (CL-RSA or CL-Bilinear), generate and verify a $ \mathsf{PoK}_{\mathsf{PCred}} $ for a disclosed attribute (e.g., \texttt{PatientID}), then sign the token with a partially blind Schnorr scheme.
    
   \item \textbf{Appointment Booking}: Using the chosen credential signature scheme, generate and verify the patient's Blind IBS signature on an \texttt{Appointment\allowbreak Schedule\allowbreak Request}, and verify the AT signature.

    \item \textbf{In-person Verification}: Using the chosen credential signature scheme, generate and verify the patient's Blind IBS signature on an \texttt{Identity\allowbreak Verification\allowbreak Request}, generate and verify $\mathsf{PoK}_{\mathsf{PCred}}$ for a disclosed attribute (e.g., \texttt{BioHash}), and verify the PT signature.
    
\end{enumerate}

\begin{table}[h]
    \centering
    \caption{NIST Recommended Equivalent Key Sizes at 128-bit Security Strength}
    \label{tab:nist_key_sizes}
    \begin{tabular}{ll}
        \hline
        \textbf{Algorithm} & \textbf{Key Size (bits)} \\
        \hline
        RSA Key Size & 3072 \\
        ECC Key Size & 256 \\
        Schnorr Key Size: $p$ (modulus) & 3072 \\
        Schnorr Key Size: $q$ (subgroup order) & 256 \\
        \hline
    \end{tabular}
\end{table}

\subsection{Simulation Setup and Results}
\label{subsec:simulation-setup-results}

The performance evaluation was conducted using a custom simulator designed to capture the core interactions of the proposed framework. The simulator was executed on a Google Cloud \texttt{e2-medium} virtual machine, equipped with two vCPUs (Intel Xeon @ 2.8 GHz), four GB of RAM, and running Ubuntu 20.04.6 LTS. The implementation was developed in Python (v3.9.7). Cryptographic parameters for the RSA and CL-RSA signature schemes were generated using the PyCryptodome library (v3.15.0) \cite{pypiClientChallenge1}, and those for CL-Bilinear using the bplib library (v0.1.5) \cite{pypiClientChallenge2}. 

Table~\ref{tab:hidm_performance_summary} summarizes the processing time measurements for the six scenarios, comparing the privacy-preserving CL-RSA and CL-Bilinear schemes and including a baseline RSA signature for the Patient Credential Issuance scenario to benchmark the overhead introduced by these privacy-preserving approaches. \textit{Payload (Pre-signature)} denotes the size of the unsigned message or request before signature generation. \textit{Full Request} and \textit{Response} represent the total message sizes exchanged during protocol execution, while \textit{Signature Payload} indicates the size of the cryptographic signature. \textit{Avg.~Time} reflects the measured processing time required to generate and verify each artifact.

\begin{table}[t!]
\centering
\caption{Performance Summary of HIDM Cryptographic Modules}
\label{tab:hidm_performance_summary}
\renewcommand{\arraystretch}{1.1}
\setlength{\tabcolsep}{3pt}
\small
\begin{tabular}{|p{2.2cm}|p{1.1cm}|p{1.15cm}|p{1.05cm}|p{1.25cm}|p{1.1cm}|}
\hline
\textbf{Scheme / Module} &
\textbf{Payload (Pre-signature) (byte)} &
\textbf{Full Request (byte)} &
\textbf{Response (byte)} &
\textbf{Signature Payload (byte)} &
\textbf{Avg. Time (ms)} \\
\hline
RSA (Credential) & 227 & -- & 727 & 384 & 100.71 \\
CL--RSA (Credential) & 227 & 2,092 & 1,154 & 387 & 555.43 \\
CL--Bilinear (Credential) & 227 & 322 & 297 & 41 & 151.29 \\
\hline
CL--RSA (AT) & -- & 4,803 & 44 & 25 & 545.67 \\
CL--Bilinear (AT) & -- & 1,983 & 42 & 177 & 112.60 \\
\hline
CL--RSA (PT) & -- & 7,399 & 107 & 79 & 690.76 \\
CL--Bilinear (PT) & -- & 3,702 & 101 & 87 & 222.80 \\
\hline
CL--RSA (Blind-IBS Key) & -- & 4,795 & 128 & 128 & 660.06 \\
CL--Bilinear (Blind-IBS Key) & -- & 817 & 128 & 128 & 120.95 \\
\hline
CL--RSA (Appointment Booking) & -- & 1,836 & 42 & 177 & 40.14 \\
CL--Bilinear (Appointment Booking) & -- & 1,436 & 37 & 177 & 38.85 \\
\hline
CL--RSA (In-Person Verification) & -- & 1,983 & 42 & 177 & 719.47 \\
CL--Bilinear (In-Person Verification) & -- & 1,831 & 54 & 177 & 272.78 \\
\hline
\end{tabular}
\end{table}

For cryptographic operations, CL--Bilinear modules exhibit a substantial reduction in processing time, typically between 40\% and 80\% compared to CL--RSA, confirming the efficiency of pairing-based constructions under equivalent 128-bit security strength. This trend is consistent across all modules. For instance, the Pseudonym-Specific Private Key (Blind-IBS) issuance achieves the largest improvement, decreasing from 660~ms to 121~ms ($\sim$82\%). Patient-side token operations are also significantly faster: Appointment Token issuance time falls from 546~ms to 113~ms ($\sim$79\%), while Pseudonym Token issuance drops from 691~ms to 223~ms ($\sim$68\%). Even routine operations such as Appointment Booking remain efficient, completing in $\sim$39~ms. 

The In-person Verification scenario remains the most computationally intensive due to the combined Blind-IBS and zero-knowledge proof verification; yet, the CL--Bilinear implementation maintains a practical latency of about 273~ms compared to 719~ms for CL--RSA ($\sim$62\% reduction). All simulation scripts, Scyther and Tamarin formal-verification models, and dataset-generation utilities used in this evaluation are publicly available at~\cite{nasif2005:healthcareIDM}.

\section{Discussion, Conclusion, and Future Work} 
\label{sec:discussion}

The contribution of this work is not the proposal of new cryptographic primitives or optimizations. Instead, the novelty lies in the system-level integration of existing, well-understood cryptographic mechanisms into a healthcare-specific identity management architecture that supports privacy, accountability, and operational continuity simultaneously.
The proposed privacy-preserving HIDM framework addresses a central challenge in health informatics: reconciling longitudinal data continuity with robust patient privacy protections. By integrating verifiable legitimacy, unlinkable pseudonymity, and conditional traceability, the design formally models and balances the long-standing tension between maintaining longitudinal health records and limiting unnecessary cross-context linkability. The MSRA analysis indicates that stakeholder requirements are systematically addressed across operational episodes. Additionally, formal cryptographic proofs and symbolic verification (via Scyther and Tamarin) demonstrate that the framework’s protocol workflows preserve unforgeability, confidentiality, and unlinkability within the defined Dolev-Yao threat model. Furthermore, the performance evaluation of core cryptographic operations indicates computational feasibility, suggesting that processing times are compatible with typical interaction latencies in clinical settings such as patient check-in and health record access within an EHR.

The framework architecture introduces specific trade-offs. The establishment of institutional legitimacy depends on the Government Health Authority (GHA), creating a degree of institutional dependency. Credential lifecycle events, such as the issuance and revocation of Legitimacy Verifiable Credentials, could be recorded in an append-only transparency log (e.g., a permissioned ledger) to mitigate this vulnerability. Such transparency mechanisms shift the model from blind reliance on a single authority toward verifiable trust, enabling independent monitoring of credential governance actions.

Conditional traceability relies on the split-knowledge design between the APC and PTA. Although simultaneous compromise or collusion is unlikely under standard trust assumptions, it remains a residual risk. Mitigation strategies include comprehensive logging of mapping-table access events and independent oversight mechanisms. Additionally, incorporating secret-sharing techniques within each authority can prevent any single administrator from reconstructing complete identity mappings without multi-party cooperation, thereby reducing insider threats.

Regarding real-world integration, the proposed HIDM framework is designed to avoid assumptions that conflict with existing healthcare IT ecosystems and to support incremental deployment alongside established infrastructures. The architecture interoperates with current EHR systems without requiring modifications to underlying clinical data formats and with minimal disruption to existing clinical workflows. Core components, such as the HRR and the Auditor Ledger, can be deployed as adjunct services, enabling gradual adoption rather than wholesale system replacement. More broadly, the modular, API-driven design positions the framework as an architectural layer for next-generation digital health ecosystems, facilitating alignment with interoperability standards such as HL7 FHIR. For instance, the Health Record Controller can expose a FHIR-compliant API for accessing patient records, allowing privacy-preserving identity management to operate alongside standardized clinical data exchange. While the administrative onboarding of patients and the establishment of government-anchored trust roots introduce deployment challenges, these design choices reflect a deliberate balance between practical implementability, interoperability requirements, and regulatory audit constraints.

Looking forward, the HIDM framework provides a foundation for extending privacy-preserving identity management across the broader healthcare ecosystem. While the present design formally models the privacy–continuity balance between patients and healthcare providers, real-world healthcare delivery also depends on health insurance organizations for billing and authorization. An important extension of this work is therefore the integration of insurance endpoints, enabling eligibility verification and billing integrity without exposing patients’ longitudinal clinical records to payers.

Another promising research direction involves integrating a zero-knowledge-based medication threshold mechanism that allows patients to demonstrate compliance with prescription limits without disclosing their complete medication histories. Such extensions would further strengthen the feasibility of privacy-preserving, auditable healthcare services that align regulatory accountability requirements with data minimization principles and patient autonomy.





\appendix         

\renewcommand{\thealgorithm}{A.\arabic{algorithm}}  

\section{Pseudonym Binding Proof}
\label{appendix:pseudonym_binding_proof}

The Pseudonym Binding Proof (PBP) is a NIZK proof that demonstrates the cryptographic linkage between a patient’s pseudonym \(\mathsf{P}_{\mathsf{patient}}\) and the healthcare-specific identifier \(\mathsf{PatientID}\) of Patient Credential.
The objective of the NIZK proof is to allow the patient to convince the PTA that the pseudonym 
\( \mathsf{P}_{\mathsf{patient}} = (P_1,\, P_2) \) is derived from their \(\mathsf{PatientID}\), without revealing either the derivation process or any additional identity attributes from Patient Credential.

\subsection*{Proof Objective}
Let \(z = e(g_1, g_2) \in G_T\), \(\mathsf{pk}_{\mathsf{patient}} = x g_2 \in G_2\), and \(h = \mathsf{HashToField}(\mathsf{PatientID}) \in \mathbb{F}_r\). 
During pseudonym generation, the patient selects a nonce \(r \in \mathbb{F}_r\), which acts as a blinding factor ensuring that different pseudonyms generated from the same \(\mathsf{PatientID}\) are unlinkable to each other.
The patient must demonstrate knowledge of the nonce \(r\) as the value that links all these components \((\mathsf{P}_{\mathsf{patient}} = (P_1,\, P_2), g_1, g_2, z, \mathsf{pk}_{\mathsf{patient}}, h)\) correctly. 

The structure of the pseudonym \( \mathsf{P}_{\mathsf{patient}} \) is given by
\[
P_1 = z^r z^{h} \;\text{or}\; \tfrac{P_1}{z^{h}} = z^{r}, 
\qquad
P_2 = r\,\mathsf{pk}_{\mathsf{patient}} = r(x g_2).
\]

\subsection*{Credential Binding Pre-Check (by PTA)}
PTA verifies the issuer's signature on the Patient Credential and obtains the attribute \( \mathsf{PatientID} \).
PTA computes \( h = \mathsf{HashToField}(\mathsf{PatientID}) \).

\subsection*{Proof Generation (by Patient)}
\begin{enumerate}
  \item Choose random scalars \( t_1, t_2 \in \mathbb{F}_r \).
  \item Commitment calculation:
    \(
      T_1 = z^{t_1} \in G_T, \quad
      T_2 = t_2\,\mathsf{pk}_{\mathsf{patient}} \in G_2
    \)
    
  \item Challenge calculation:
    \[
    c = \mathsf{HashToField}\!\big(\mathsf{P}_{\mathsf{patient}} \parallel T_1 \parallel T_2 \parallel \mathsf{pk}_{\mathsf{patient}} \parallel h\big) \in \mathbb{F}_r
    \]

  \item Response calculation:
    \(
      s_1 = t_1 + c r, \quad s_2 = t_2 + c r \quad \in \mathbb{F}_r
    \)
  \item Proof calculation:
    \(
      \pi = (T_1, T_2, c, s_1, s_2)
    \)
\end{enumerate}

\subsection*{Proof Verification (by PTA)}
\begin{enumerate}

  \item Recompute challenge
    \[
    c' = \mathsf{HashToField}\!\big(\mathsf{P}_{\mathsf{patient}} \parallel T_1 \parallel T_2 
    \parallel \mathsf{pk}_{\mathsf{patient}} \parallel h\big) \in \mathbb{F}_r.
    \]

  \item Check in \(G_T\) (multiplicative):
    \[
    z^{s_1} \stackrel{?}{=} T_1 \cdot \Big(\tfrac{P_1}{z^{h}}\Big)^{c'}
    \]
    
  \item Check in \(G_2\) (additive):
    \[
    s_2\,\mathsf{pk}_{\mathsf{patient}} \stackrel{?}{=} T_2 + c' P_2
    \]

  \item Accept iff both equations hold and the credential pre-check passed.
  
\end{enumerate}

\subsection*{Correctness}
The verification holds since 
\( z^{s_1} = z^{t_1 + c r} = z^{t_1}(z^r)^c = T_1 (z^r)^c = T_1 \big(\tfrac{P_1}{z^{h}}\big)^c \)
and 
\( s_2\,\mathsf{pk}_{\mathsf{patient}} = (t_2 + c r)\,\mathsf{pk}_{\mathsf{patient}} = T_2 + c (r\,\mathsf{pk}_{\mathsf{patient}}) = T_2 + c P_2. \)


\bibliographystyle{IEEEtran}
\bibliography{IEEEfull}

\end{document}